%% file: main.tex
\DocumentMetadata{}  
\documentclass[sigconf,balance=false]{acmart}
\usepackage{popets}

\setcopyright{popets}
\copyrightyear{2025}

\acmYear{2025}
\acmVolume{2025}
\acmNumber{1}
\acmDOI{10.56553/popets-2025-0030}
\acmISBN{}
\acmConference{Proceedings on Privacy Enhancing Technologies}
\startPage{553}
\usepackage{amsfonts}  
\usepackage{amsmath}  
\usepackage{amsthm}  
    \theoremstyle{definition}
        \newtheorem{corollary}{Corollary}
            
        \newtheorem{definition}{Definition}
            
        \newtheorem{lemma}{Lemma}

        \newtheorem{property}{Property}
            
        \newtheorem{theorem}{Theorem}
            
    \theoremstyle{remark}
        \newtheorem{remark}{Remark}
            
\usepackage{mathtools}  
    \DeclareMathOperator{\girth}{girth}
    \DeclareMathOperator{\rref}{rref}
    
    \DeclarePairedDelimiter\abs{\lvert}{\rvert}
    \DeclarePairedDelimiter\range{\llbracket}{\rrbracket}
\usepackage{nicematrix}  
\usepackage{xfrac}  

\usepackage{subcaption}  
    \newsavebox{\imagebox}  
\usepackage{tikz}  
    \usetikzlibrary{fit}  
    \usetikzlibrary{positioning}  

\usepackage{csquotes}  
\usepackage[inline]{enumitem}  

\begin{document}
    \pdfstringdefDisableCommands{\def\\{ }}
    \title[Topology-Based Reconstruction Prevention for Decentralised Learning]{Topology-Based Reconstruction Prevention for\\Decentralised Learning}

    \author{Florine W.\ Dekker}
    \orcid{0000-0002-0506-7365}
    \affiliation{%
        \institution{Delft University of Technology}
        \city{Delft}
        \country{Netherlands}}
    \email{f.w.dekker@tudelft.nl}

    \author{Zekeriya Erkin}
    \orcid{0000-0001-8932-4703}
    \affiliation{%
        \institution{Delft University of Technology}
        \city{Delft}
        \country{Netherlands}}
    \email{z.erkin@tudelft.nl}

    \author{Mauro Conti}
    \orcid{0000-0002-3612-1934}
    \affiliation{%
        \institution{Università di Padova}
        \city{Padua}
        \country{Italy}}
    \affiliation{%
        \institution{Delft University of Technology}
        \city{Delft}
        \country{Netherlands}}
    \email{mauro.conti@unipd.it}

    \input{00-abstract}
    \maketitle

    \input{01-introduction}
    \input{02-related-work}
    \input{03-preliminaries}
    \input{04-reconstruction-in-mpc}
    \input{05-resistence-by-girth}
    \input{06-conclusion}

\begin{acks}
    This paper is partly supported by the European Union’s Horizon Europe research and innovation program under grant agreement No. 101094901, the SEPTON project (Security Protection Tools for Networked Medical Devices).
\end{acks}

    \bibliographystyle{ACM-Reference-Format}
    \hbadness=10000\hfuzz=10000pt\bibliography{main}
\end{document}

%% file: 00-abstract.tex
\begin{abstract}
    Decentralised learning has recently gained traction as an alternative to federated learning in which both data and
    coordination are distributed over its users.
    To preserve the confidentiality of users' data, decentralised learning relies on differential privacy, multi-party
    computation, or a combination thereof.
    However, running multiple privacy-preserving summations in sequence may allow adversaries to perform reconstruction
    attacks.
    Unfortunately, current reconstruction countermeasures either cannot trivially be adapted to the distributed setting,
    or add excessive amounts of noise.

    In this work, we first show that passive honest-but-curious adversaries can infer other users' private data after
    several privacy-preserving summations.
    For example, in subgraphs with 18~users, we show that only three passive honest-but-curious adversaries succeed at
    reconstructing private data 11.0\% of the time, requiring an average of 8.8~summations per adversary.
    The success rate depends only on the adversaries' direct neighbourhood, and is independent of the size of the full
    network.
    We consider weak adversaries that do not control the graph topology, cannot exploit the inner workings of the
    summation protocol, and do not have auxiliary knowledge;
    and show that these adversaries can still infer private data.

    We develop a mathematical understanding of how reconstruction relates to topology and propose the first
    topology-based decentralised defence against reconstruction attacks.
    Specifically, we show that reconstruction requires a number of adversaries linear in the length of the network's
    shortest cycle.
    Consequently, exact reconstruction attacks over privacy-preserving summations are impossible in acyclic networks.

    Our work is a stepping stone for a formal theory of topology-based decentralised reconstruction defences.
    Such a theory would generalise our countermeasure beyond summation, define confidentiality in terms of entropy, and
    describe the interactions with (topology-aware) differential privacy.
\end{abstract}

%% file: 01-introduction.tex
\section{Introduction}\label{sec:introduction}
Machine learning is used in a wide array of systems, including smartwatches~\cite{Weiss2016},
predictive text~\cite{Bonawitz2017}, and malware detection~\cite{Rieck2011}.
These systems require access to large amounts of reliable data in order to function accurately.
In practice, the necessary data usually exist, but are distributed over many data owners.
The naive approach for data collection is to have the data owners send their data to a central server, which trains a
machine learning model on these data before deploying it.
However, sharing private data may result in misuse, for example in the form of targeted advertising or harassment.
In an age of increasing privacy awareness, data owners may be reluctant to share their data, threatening the viability
of data-intensive machine learning applications.

The emerging field of federated learning, first formalised in~\cite{McMahan2017}, addresses these privacy issues by
distributing the training process over the data owners.
Instead of submitting their data, each data owner first trains a machine learning model on their local data and then
submits this model to a central server.
This central server, called the aggregator, uses a privacy-preserving summation protocol to combine the received models
into a single global model.
The central server then sends back the global model to the data owners, who apply another round of training, repeating
the entire process until the global model has converged.

A significant drawback of classical federated learning is that communication is a bottleneck, scaling
quadratically~\cite{Bonawitz2017} or poly-logarithmically~\cite{Bell2020} in the number of users.
Decentralised learning, a variant of federated learning~\cite{Kairouz2019}, removes this bottleneck by distributing both
the data and the coordination between users.
Training happens in a peer-to-peer fashion, with users exchanging information only with their direct neighbours.
This significantly reduces the communication complexity~\cite{Lian2017}, allowing for cost-effective deployments without
a central server.
Furthermore, because communication is local, it becomes much harder for adversaries to observe the full
network~\cite{Troncoso2017}.

Recently, there has been increased interest in decentralised learning.
Though some works do not consider privacy~\cite{Lian2017, Tang2018, Zantedeschi2020}, many other works do.
Some of these works~\cite{Vanhaesebrouck2017, Bellet2018, Zantedeschi2020} consider algorithms in which nodes are
randomly selected to calculate updates, and protect the private data underlying the models using differential privacy.
That is, they apply carefully calibrated random noise to the calculated gradients before sharing them with others.
A slight variation of this is to use a random walk through the graph to determine the order in which updates
occur~\cite{Cyffers2022}.
There are also works~\cite{Chen2018, Qu2020, Schmid2020} that use blockchains to facilitate the communication and
coordination between nodes, and then similarly use differential privacy.
Finally, instead of differential privacy, some works utilise multi-party
computation~\cite{Danner2018, Kanagavelu2020, Tran2021}, which does not give noisy results, but has higher
computational costs.

A common thread in these works is that they apparently assume that if a single summation is secure, then the protocol
remains secure after multiple summations.
However, this requires further scrutiny, as combining information from multiple rounds may reveal previously hidden
information.
For example, given private records~$A$, $B$, and~$C$, and a privacy-preserving summation protocol, an adversary could
separately query~$A + B$, then~$B + C$, and finally~$A + C$, and use a linear algebra solver to learn all three private
records.
To defend against such attacks, one must prevent sequences of queries that would reveal private data.
Naive restrictions, such as requiring a minimum number of included records per query, are insufficient:
The adversary could still first query the sum of all models and then query the sum of all models except one, allowing
them to reconstruct the excluded model.
As such, designing proper countermeasures requires a formal theory.

Extracting data from output traces is known as a reconstruction attack, which has its roots in the theory of statistical
disclosure~\cite{Fellegi1972}.
Many defences have been proposed since the 1970s, including query auditing~\cite{Chin1982},
perturbation~\cite{Dwork2006a}, and random sampling~\cite{Denning1980}.
However, these works assume either a central database, or otherwise assume a central arbiter that determines which
queries are allowed.
In decentralised learning, there is no clear leader who can be trusted to audit queries.
Instead, decentralised learning requires a decentralised solution.
Apart from works on perturbation, to the best of our knowledge, only \citeauthor{DaSilva2004}~\cite{DaSilva2004} have
considered reconstruction attacks in peer-to-peer networks, but their work applies only to distributed clustering, and
does not propose any countermeasures.
When considering perturbation, naively applying user-level differential privacy in a distributed setting results in
linearly-scaling noise, severely reducing the protocol's utility~\cite{Dwork2014, Zheng2017, Cyffers2022}.
Intuitively, utility can be increased while retaining the level of privacy by correlating noise by
topology~\cite{Dwork2006a}, but to the best of our knowledge only a few works have done this.
\citeauthor{Guo2022}~\cite{Guo2022} reduce noise based on the mutual overlaps of neighbours' neighbourhoods, but do not
consider time-series correlations.
\citeauthor{Cyffers2022}~\cite{Cyffers2022} observe that data sensitivity decreases as mutual node distances increase,
but their solution does not scale well under collusion.

In this work, we analyse reconstruction attacks performed by colluding adversaries in peer-to-peer networks.
We model the network after decentralised learning, though our analysis is sufficiently generic to describe a sequence of
summations in any environment.
Summation is a simple protocol, but is sufficient to implement many of the aforementioned decentralised learning
protocols, in addition to smart metering~\cite{Garcia2011} and even principal component analysis, singular-value
decomposition, and decision tree classifications~\cite{Blum2005}.
We assume a set of nodes, each with a private datum that changes over time, and allow privacy-preserving summation over
one's direct neighbours.
We do not consider auxiliary knowledge;
see \autoref{subsubsec:preliminaries:assumptions-and-notation:adversarial-model} and \cite{Cormode2013, Clifton2013} for
a detailed discussion on the real-world applicability of this model.
We then formalise the relation between reconstruction and network topology, and prove that exact reconstruction attacks
are impossible in a specific class of topologies.

Concretely, we begin by showing that reconstruction attacks are practical, and that, in random peer-to-peer subgraphs,
three honest-but-curious adversaries with 15~neighbours succeed in finding at least one neighbours' private datum with
an 11.0\% success rate, requiring an average of only 8.8~rounds per adversary.
The success rate is independent of the size of the full network;
it depends only on the adversaries' local neighbourhood.
We then show that the success rate depends on the connectivity of the network rather than its size.
Specifically, we show that reconstruction corresponds to cycles in the graph:
If the graph's shortest cycle has length~$2k$, then reconstruction never succeeds if there are fewer than
$k$~adversaries.
Finally, we briefly evaluate the impact of increasing girth on the convergence of a distributed averaging protocol, and
find that while more rounds are required to achieve convergence in all graphs, dense graphs require fewer rounds than
sparse graphs do when both are \enquote{stretched} to higher girths.

To the best of our knowledge, our work is the first to propose a topology-based decentralised defence to reconstruction
attacks.
We show that restricting how summations may be composed makes it impossible to reconstruct private data.
We must therefore assume that adversaries do not have auxiliary knowledge, as restrictions on summations cannot be
guaranteed otherwise.
With the ultimate goal of developing a general theory of structured composition as a distributed reconstruction
countermeasure, future work may include finding a condition that is not only sufficient (as seen in this work) but also
necessary for reconstruction, generalising these countermeasures to operations beyond summation, stronger notions of
privacy rooted in information theory, and investigating the interactions with (topology-aware) differentially private
noise.

The remainder of this paper is structured as follows.
In \autoref{sec:related-work}, we discuss related work.
In \autoref{sec:preliminaries}, we describe the preliminaries:
We explain basic primitives, formalise our assumptions, and introduce our notation.
In \autoref{sec:reconstruction-in-multi-party-summation}, we formally describe reconstruction attacks, and show that the
attack is feasible.
In \autoref{sec:resistance-by-girth}, we prove that the success rate of the reconstruction attack depends on the graph's
girth, and investigate how girth affects application performance.
Finally, in \autoref{sec:conclusion}, we present our conclusions.

%% file: 02-related-work.tex
\section{Related Work}\label{sec:related-work}
In this work we propose a decentralised reconstruction countermeasure for privacy-preserving summation with dynamic
data.
To the best of our knowledge, this exact problem has not been treated in literature before.
Therefore, in this section, we consider related works from various fields, and describe their similarities and
differences.

\subsection{Reconstruction Attacks}\label{subsec:related-work:reconstruction-attack}
Consider a database that users can query for statistical information.
For example, in a database with employee records, users can query for the sum of salaries of all PhD~students.
Naturally, the database must ensure that users cannot learn individual employees' salaries.
A naive defence would disallow queries over single records, but a cleverly chosen sequence of queries may still allow
the user to reconstruct private data.
For example, the user could query the sum of salaries of all employees, and the sum of salaries of all employees except
Jay~Doe, and reconstruct Jay~Doe's salary from that.

The attack described above is known under various names:
\textit{statistical disclosure}%
\footnote{Confusingly, the term \enquote{statistical disclosure attack} is also a separate attack in peer-to-peer
literature~\cite{Danezis2003}, but this is an unrelated attack on anonymity rather than confidentiality.},
the \textit{inference problem}, and the \textit{reconstruction attack}.
It has been the subject of research since at least the 1970s~\cite{Fellegi1972}, originally in the context of releasing
census statistics.
Since then, many reconstruction defences have been proposed, including random sampling~\cite{Denning1980}, query
auditing~\cite{Chin1982}, and perturbation~\cite{Dwork2006a}.

Most related to our research question are those works that consider sum queries only.
\citeauthor{Chin1978}~\cite{Chin1978} studies summation query graphs to determine the exact conditions
under which disclosure occurs.
However, his analysis is limited to queries that are over exactly two records each, and cannot easily be generalised.
\citeauthor{Wang2002}~\cite{Wang2002} allow queries over more than two records.
The authors propose cardinality-based criteria for determining whether reconstruction is possible, and create a
whitelist of all summations that can be performed without causing undesired reconstruction.

All aforementioned solutions consider a single trusted database or auditor, making them unsuitable for peer-to-peer
protocols, in which the data are spread over many users.
Except for perturbation-based techniques, there are very few works that consider reconstruction defences in peer-to-peer
settings.
In their study on reconstruction attacks in distributed environments, \citeauthor{Jebali2020}~\cite{Jebali2020} note
only the work by \citeauthor{DaSilva2004}~\cite{DaSilva2004} when discussing peer-to-peer solutions, but the latter
applies only to distributed clustering, and does not propose any countermeasures.

Perturbation, on the other hand, has been studied in more detail.
Probably the most popular perturbation mechanism for the decentralised setting is local differential
privacy~\cite{Warner1965, Evfimievski2003, Kasiviswanathan2008}, a variation of differential privacy~\cite{Dwork2006a}.
With this technique, when a query is performed over some set of nodes, each node adds a small amount of noise such that
the aggregate is relatively accurate, but reconstruction remains impossible even after multiple queries.
Various fully-decentralised learning protocols use local differential privacy to allow learning a shared machine
learning model without revealing users' private datasets~\cite{Vanhaesebrouck2017, Bellet2018, Zantedeschi2020}.
However, the perturbation is calibrated to protect individual records in users' private datasets, rather than protecting
users' entire datasets.
As a result, these works are potentially vulnerable to inversion attacks~\cite{Hitaj2017, Wang2019}.
The level of noise can be increased, but this severely impacts utility~\cite{Zheng2017, Cyffers2022}.
Intuitively, noise can be made more \enquote{efficient} by exploiting correlations between users' data~\cite{Dwork2014},
which, in peer-to-peer networks, amounts to calibrating noise to the topology.
To the best of our knowledge only a few works have done this.
\citeauthor{Guo2022}~\cite{Guo2022} reduce noise based on the mutual overlaps of neighbours' neighbourhoods, but do not
consider time-series correlations.
\citeauthor{Cyffers2022}~\cite{Cyffers2022} observe that data sensitivity decreases as mutual node distances increase,
but their solution does not scale well when adversaries collude.

\subsection{Multi-Party Computation}\label{subsec:related-work:multi-party-computation}
In secure multi-party computation, composability~\cite{Lindell2003} is the property of a cryptographic scheme that no
additional leakage occurs when it is invoked multiple times, with varying parties, combined with other schemes, and so
on.
There are numerous frameworks to model composability, including universal composability~\cite{Canetti2001}, constructive
composability~\cite{Maurer2012}, and reactive simulatability~\cite{Backes2007}.

Composability solves a different issue than the one posed in this work.
While composability ensures nothing leaks beyond what can be inferred from the outputs, our work is concerned exactly
with that which can be inferred from the outputs.
Composability does not help when the desired output (implicitly) reveals private data.

In secure multi-party computation literature, this difference is occasionally acknowledged.
For example, \citeauthor{Bogdanov2014}~\cite{Bogdanov2014} note that \enquote{the composition of ideal functionalities
is no longer an ideal functionality}, and, before them, \citeauthor{Yang2010}~\cite{Yang2010} made a similar
observation.
There are more works that consider this difference, but, to the best of our knowledge, these works all resolve the issue
by removing or protecting intermediate values, but do not consider protocols which desire intermediate values, and even
then do not consider that reconstruction attacks may be possible after multiple instantiations of the protocol.
An exception is the work by \citeauthor{Dekker2021}~\cite{Dekker2021}, which releases intermediate values in a
structured manner such that it is not possible to reconstruct all users' values.
However, the authors do not prove (or disprove) that it is impossible to find a \textit{single} user's value.

%% file: 03-preliminaries.tex
\section{Preliminaries}\label{sec:preliminaries}
We briefly explain some basics on privacy-preserving summation in
\autoref{subsec:preliminaries:privacy-preserving-summation} and on bipartite graphs in
\autoref{subsec:preliminaries:bipartite-graphs}.
After that, we formulate our assumptions and define our notation in
\autoref{subsec:preliminaries:assumptions-and-notation}.

\subsection{Privacy-Preserving Summation}\label{subsec:preliminaries:privacy-preserving-summation}
Privacy-preserving summation is a special case of multi-party computation in which an aggregator calculates the sum of
users' private values without learning the users' individual values.
In this work, we consider privacy-preserving summation to be an information-theoretically secure black-box that reveals
only the identities and the sum of the variables.

\subsection{Bipartite Graphs}\label{subsec:preliminaries:bipartite-graphs}
A bipartite graph~$H = (U, V, E)$ is a graph with nodes $U \cup V$ and edges~$E$, subject to
$U \cap V = \emptyset$ and $\forall (u, v) \in E : u \in U \Leftrightarrow v \in V$.

Furthermore, a bipartite graph~$H = (U, V, E)$ can be described by a biadjacency
matrix~$A \in \{0, 1\}^{\abs{U} \times \abs{V}}$, where
$\forall 0 \leq u < \abs{U}, 0 \leq v < \abs{V} : A_{u, v} = 1 \Leftrightarrow (U_u, V_v) \in E$.

In this work, all graphs are undirected.

\subsection{Assumptions and Notation}\label{subsec:preliminaries:assumptions-and-notation}
The underlying models and assumptions in this work are based on those seen in the decentralised learning
literature~\cite{Danner2018, Bellet2018, Zantedeschi2020}, but are especially close to the work by
\citeauthor{Vanhaesebrouck2017}~\cite{Vanhaesebrouck2017}.

In general, we denote the first element of a vector~$v$ by~$v_0$, the first row of a matrix~$A$ by~$A_0$, the range of
integers~$\{ 0 \ldots n - 1 \}$ by~$\range{n}$, and the number of elements in a collection~$S$ by~$\abs{S}$.

\subsubsection{User data and objectives}
\label{subsubsec:preliminaries:assumptions-and-notation:user-data-and-objectives}
Consider a system of $n$~users~$V$, each with a private datum.
Each datum is dynamic;
it changes each time the user initiates a round and incorporates new knowledge from their neighbours.
(We describe the time model in \autoref{subsubsec:preliminaries:assumptions-and-notation:time-model}.)
Each datum can be a vector of values, though for simplicity we assume scalar values in our notation.
Examples of dynamic data are power consumption, GPS~coordinates, and machine learning models.
In round~$t$, the data of user~$i \in \range{n}$ is denoted~$\theta_{i, t}$.

The users want to compute some function over their data without revealing their data to others.
Each user regularly runs a privacy-preserving summation protocol to find the sum of their direct neighbours' private
data.
This sum can be used for principal component analysis, singular-value decomposition, or distributed gradient descent,
for example.

\subsubsection{Network model}\label{subsubsec:preliminaries:assumptions-and-notation:network-model}
Users communicate with each other in a peer-to-peer network.
This can be a physical network, for example based on Bluetooth or Wi-Fi~Direct, or an overlay network, in which users
are connected through the Internet.
We model the network as an undirected, self-loopless, static graph~$G = (V, E)$ in which each node represents a user.
(We consider graphs with dynamic edges in \autoref{subsec:resistance-by-girth:dynamic-edges}.)
The direct neighbours of a node~$v \in V$ are denoted~$N_G(v)$, and for any set of users~$U \subseteq V$ we define their
shared neighbours~$N_G(U) \coloneq \bigcup_{u \in U} N_G(u) \setminus U$.
The network topology is not private;
in fact, users know who their direct neighbours are.
Users may run a privacy-preserving summation protocol to learn the sum of their direct neighbours' private values.

\subsubsection{Adversarial model}\label{subsubsec:preliminaries:assumptions-and-notation:adversarial-model}
We assume all $n$~users~$V$ are honest-but-curious.
That is, all users honestly follow the protocol, but may attempt to obtain other users' private data by operating on the
data obtained in the protocol in any way they see fit.
Additionally, $k$~users~$C \subseteq V$ may collude with each other, but we require that each adversary has either zero
or at least two non-adversary neighbours, as retrieving private data is trivial otherwise.
We give an example of a valid set of adversaries in
\autoref{fig:preliminaries:assumptions-and-notation:adversarial-model:example}.
Colluding users are still honest-but-curious, so their collusion is limited to sharing information outside the
protocol.
While excluding all actively malicious behaviour is a strenuous assumption in practice, we argue that the challenges in
the honest-but-curious model are already sufficiently interesting to warrant investigation.
We leave stronger notions of adversarial behaviour to future work;
see also \autoref{sec:conclusion}.

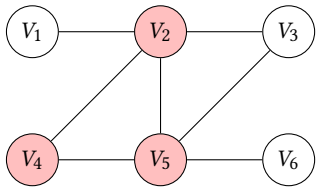
\begin{figure}[htb]
    \centering
    \input{fig/adversarial-model-example.tikz}
    \caption{
        A network with 6~users~$V$.
        The adversaries~$C = \{ V_2, V_4, V_5 \}$ are shaded.
        Removing edge $(V_2, V_3)$ would violate our requirements, as adversary~$V_2$ would have exactly one
        non-adversary neighbour.
    }
    \Description{
        An undirected graph with six nodes, numbered one through six.
        Nodes two, four, and five are adversaries, and nodes one, three, and six are non-adversaries.
        Node two, an adversary, is connected to adversary nodes four and five, and to non-adversary nodes one and three.
        Node four, an adversary, is connected to adversary nodes two and five, and is not connected to any non-adversary
        nodes.
        Node five, an adversary, is connected to adversary nodes two and four, and to non-adversary nodes three and six.
    }
    \label{fig:preliminaries:assumptions-and-notation:adversarial-model:example}
\end{figure}

Finally, we assume that adversaries do not possess auxiliary knowledge.
That is, we aim for syntactic privacy~\cite{Clifton2013}, of which the privacy guarantees do not compose trivially with
those of other protocols using the same private data.
Syntactic privacy is suitable when high utility is desired and participants have some level of mutual
trust~\cite{Cormode2013, Clifton2013}.
Moreover, prescribing a syntax on the data is inherent to this work's goal of establishing an interpretable relation
between privacy and topology.
We note that syntactic privacy does not preclude the use of semantic protections such as differential privacy, though
the investigation of that combination is out of scope for this work.
See \cite{Cormode2013, Clifton2013} for a detailed discussion of the subject.

\subsubsection{Time model}\label{subsubsec:preliminaries:assumptions-and-notation:time-model}
We work in the asynchronous time model~\cite{Boyd2006}, in which a global clock ticks whenever a user wakes up and
performs some work.
Equivalently, each user has their own clock ticking at the speed of a rate-1 Poisson process;
when a user's clock ticks, that user wakes up.
We denote the current global round number by~$t$ (for \enquote{time}).

%% file: fig/adversarial-model-example.tikz
\begin{tikzpicture}
[node/.style={draw, circle}, adversary/.style={fill=pink}]
    \node[node,          ]              (V1) {$V_1$};
    \node[node, adversary, right=of V1] (V2) {$V_2$};
    \node[node,            right=of V2] (V3) {$V_3$};
    \node[node, adversary, below=of V1] (V4) {$V_4$};
    \node[node, adversary, right=of V4] (V5) {$V_5$};
    \node[node,            right=of V5] (V6) {$V_6$};

    \draw (V1) -- (V2);
    \draw (V2) -- (V3);
    \draw (V2) -- (V4);
    \draw (V2) -- (V5);
    \draw (V3) -- (V5);
    \draw (V4) -- (V5);
    \draw (V5) -- (V6);
\end{tikzpicture}

%% file: 04-reconstruction-in-mpc.tex
\section{Reconstruction in Multi-party Summation}\label{sec:reconstruction-in-multi-party-summation}
In this section we formally define reconstruction attacks in privacy-preserving multi-party dynamic-data summation, and
experimentally verify that this attack is feasible.
Adversaries passively record the summations they obtain throughout the protocol.
Because adversaries know which users are included in which summation, they obtain a system of linear equations.
Even if the system has no global solutions, adversaries may still learn the private data of some users.

In \autoref{subsec:reconstruction-in-multi-party-summation:reconstruction}, we informally explain reconstruction attacks
with examples.
In \autoref{subsec:reconstruction-in-multi-party-summation:obtained-adversarial-knowledge}, we give an exact definition
of the adversaries' knowledge.
In \autoref{subsec:reconstruction-in-multi-party-summation:the-attack}, we formally define reconstruction on multi-party
dynamic-data summation.
In \autoref{subsec:reconstruction-in-multi-party-summation:feasibility}, we experimentally verify the feasibility and
success rate of reconstruction attacks on random graphs.

\subsection{Introduction to Reconstruction Attacks}\label{subsec:reconstruction-in-multi-party-summation:reconstruction}
For this brief introduction, we use somewhat informal notation.
We formally define our notation in
\autoref{subsec:reconstruction-in-multi-party-summation:obtained-adversarial-knowledge}.

\paragraph{A small example}
Consider a graph~$G = (V, E)$ with users~$V$ and a set of $k$~adversaries $C \subseteq V$.
If a single adversary~$c \in C$ sums their neighbours' values, they learn a linear equation~$\Theta_c$ over the private
values~$\theta$ of neighbours~$N_G(c)$.
If multiple adversaries~$C$ collude, they share a \textit{system} of linear equations~$A\theta = \Theta$ over the
private values~$\theta$ of~$N_G(C)$.
If the system of linear equations has a solution, then the adversaries are able to calculate all observed users' private
values using linear combinations of the system's rows.
For example, given adversaries $A$, $B$, and $C$ with observations
\begin{equation}
    \label{eq:reconstruction-in-multi-party-summation:introduction:solvable-system-observations}
    \begin{array}{ccccccl}
        \theta_1 & + & \theta_2 &   &          & = & \Theta_A,             \\
        \theta_1 &   &          & + & \theta_3 & = & \Theta_B,\ \text{and} \\
        &   & \theta_2 & + & \theta_3 & = & \Theta_C,
    \end{array}
\end{equation}
this is equivalent to the system of linear equations
\begin{equation}
    \label{eq:reconstruction-in-multi-party-summation:introduction:solvable-system-matrix}
    \begin{bmatrix}
        1 & 1 & 0 \\
        1 & 0 & 1 \\
        0 & 1 & 1
    \end{bmatrix}
    \theta
    =
    \Theta.
\end{equation}
Since this system is full rank, adversaries can calculate
\begin{align}
    \label{eq:reconstruction-in-multi-party-summation:introduction:solvable-system-solution}
    \theta_1 &= \frac{\Theta_A + \Theta_B + \Theta_C}{2} - \Theta_C, \\
    \theta_2 &= \Theta_A - \theta_1,\ \text{and} \\
    \theta_3 &= \Theta_B - \theta_1.
\end{align}
For example, if $\Theta_A = 7$, $\Theta_B = 13$, and $\Theta_C = 8$, the adversaries know with certainty that
$\theta_1 = 6$, $\theta_2 = 1$, and $\theta_3 = 7$.
Observe that this works even if each individual summation in
\autoref{eq:reconstruction-in-multi-party-summation:introduction:solvable-system-observations} is
information-theoretically secure.

\paragraph{Partial solutions}
If the system is rank-deficient, no unique solution exists, but the system may still have partial solutions.
That is, even if a system has infinitely many possible solutions, it may be the case that some variables have the same
value in all solutions.
Even a single user's private value being leaked is a major issue for any privacy-preserving protocol.
Consider, for example, the adversarial knowledge consisting of
\begin{equation}
    \label{eq:reconstruction-in-multi-party-summation:introduction:rank-deficient-observations}
    \begin{array}{ccccccl}
        \theta_1 & + & \theta_2 & + & \theta_3 & = & \Theta_A\ \text{and} \\
        \theta_1 & + & \theta_2 &   &          & = & \Theta_B.
    \end{array}
\end{equation}
Even though there is no unique solution, all solutions have the same value for~$\theta_3$, calculated
as~$\theta_3 = \Theta_A - \Theta_B$.

The case of \autoref{eq:reconstruction-in-multi-party-summation:introduction:rank-deficient-observations} is trivial
because~$\Theta_B$ is the sum over a subset of~$\Theta_A$.
However, there are also rank-deficient systems in which no summation is a subset of another:
\begin{equation}
    \label{eq:reconstruction-in-multi-party-summation:introduction:rank-deficient-observations-non-trivial}
    \begin{array}{ccccccccl}
        \theta_1 & + & \theta_2 & + & \theta_3 &   &          & = & \Theta_A,             \\
        \theta_1 & + & \theta_2 &   &          & + & \theta_4 & = & \Theta_B,\ \text{and} \\
        &   &          &   & \theta_3 & + & \theta_4 & = & \Theta_C.
    \end{array}
\end{equation}
This system, too, has an infinite number of solutions, but each possible solution has the same values
\begin{align}
    \label{eq:reconstruction-in-multi-party-summation:introduction:solvable-system-solution-2}
    \theta_3 &= \frac{\Theta_A + \Theta_C - \Theta_B}{2}\ \text{and} \\[.6em]
    \theta_4 &= \frac{\Theta_B + \Theta_C - \Theta_A}{2}.
\end{align}

\paragraph{Time dimension}
The above examples do not take into account that users' data change over time.
To model dynamic data, first recall from
\autoref{subsubsec:preliminaries:assumptions-and-notation:user-data-and-objectives} that users update their values only
after initiating a summation.
Since each update requires an interactive summation, users implicitly inform their neighbours whenever they update;
and since each update represents the introduction of a new unknown value to~$\theta$, adversaries can represent an
update by adding a new column to their adversarial knowledge.
If a user updates their value multiple times before being observed by an adversary, the adversaries treat this as a
single update.

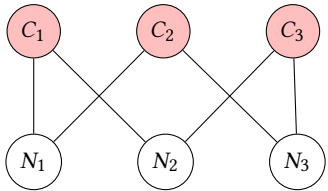
\begin{figure}[!ht]
    \centering
    \input{fig/adversarial-knowledge-example.tikz}
    \caption{%
        Example graph~$G$ with adversaries~$C = \{ C_1, C_2, C_3 \}$ (shaded) and
        non-adversaries~$N = N_G(C) = \{ N_1, N_2, N_3 \}$.
    }
    \Description{
        The undirected two-two-biregular bipartite graph with six nodes.
        The nodes on one side have labels starting with the letter C, and the nodes on the other side have labels
        starting with the letter N.
        All six nodes have two edges each.
    }
    \label{fig:reconstruction-in-multi-party-summation:introduction:example}
\end{figure}

To give an example, consider adversaries~$C$ and their neighbours~$N_G(C)$ in
\autoref{fig:reconstruction-in-multi-party-summation:introduction:example}.
Say that initially adversaries~$C_1$ and~$C_2$ run their summations, learning
\begin{equation}
    \label{eq:reconstruction-in-multi-party-summation:introduction:updating-values-observations:1}
    \begin{bNiceArray}{c|c|c}[margin]
        1 & 1 & 0 \\
        1 & 0 & 1
    \end{bNiceArray}
    \theta
    =
    \Theta.
\end{equation}
The added vertical lines group the columns per non-adversarial user.
Next, say that user~$N_1$ updates their private value.
This is noticed by the adversaries, who insert a new column into their system of equations.
If user~$C_1$ then does another summation (which includes user~$N_1$'s new value), the adversaries know
\vspace{6mm}  
\begin{equation}
    \label{eq:reconstruction-in-multi-party-summation:introduction:updating-values-observations:2}
    \begin{bNiceArray}{cc|c|c}[margin]
        1 & 0 & 1 & 0 \\
        1 & 0 & 0 & 1 \\
        0 & 1 & 1 & 0 \\
        \CodeAfter
        \OverBrace[shorten, yshift=3pt]{1-1}{3-2}{N_1}
    \end{bNiceArray}
    \theta
    =
    \Theta.
\end{equation}
The last row represents adversary~$C_1$'s new summation, and the second column represents user~$N_1$'s new value.
Finally, if users~$N_1$, $N_2$, and~$C_1$ subsequently update (in that order), then users~$N_1$ and~$N_2$ each get a new
column, and~$C_1$'s update adds a new row, giving
\vspace{6mm}  
\begin{equation}
    \label{eq:reconstruction-in-multi-party-summation:introduction:updating-values-observations:3}
    \begin{bNiceArray}{ccc|cc|c}[margin]
        1 & 0 & 0 & 1 & 0 & 0 \\
        1 & 0 & 0 & 0 & 0 & 1 \\
        0 & 1 & 0 & 1 & 0 & 0 \\
        0 & 0 & 1 & 0 & 1 & 0 \\
        \CodeAfter
        \OverBrace[shorten, yshift=3pt]{1-1}{4-3}{N_1}
        \OverBrace[shorten, yshift=3pt]{1-4}{4-5}{N_2}
    \end{bNiceArray}
    \theta
    =
    \Theta.
\end{equation}

In the remainder of this work, to simplify notation, we will always assign the same number of columns~$t$ to each user.

\paragraph{Observations}
Before we give a formal definition of reconstruction attacks, we make two observations:
\begin{enumerate}
    \item
    Reconstruction does not rely on weaknesses in the summation algorithm;
    \textbf{reconstruction works even if summation is done by a trusted third party}.
    Instead, reconstruction relies only on the summation revealing both the identities of included variables and the sum
    of those variables.

    \item
    Reconstruction is independent of how users update their private values, and works even if users update their models
    in random ways or multiple times.
    \textbf{Reconstruction works because adversaries observe multiple summations with at least one unchanged value, and
    know how the summations are related.}
\end{enumerate}

\subsection{Obtained Adversarial Knowledge}
\label{subsec:reconstruction-in-multi-party-summation:obtained-adversarial-knowledge}
We give a formal description of adversarial knowledge, which is the system of linear equations that adversaries obtain
in a privacy-preserving multi-party dynamic-data summation protocol, and observe two important properties.

Let~$G = (V, E)$ be an undirected graph, let~$C \subseteq V$ be a collusion of $k$~adversaries,
let~$n \coloneq \abs{N_G(C)}$, and let~$t \in \mathbb{N}$ be the number of summations performed by~$C$.
\begin{definition}[Adversarial knowledge]
    \label{def:reconstruction-in-multi-party-summation:obtained-adversarial-knowledge:definition}
    The adversarial knowledge over $t$~summations by~$C$ is a consistent system of linear equations~$A\theta = \Theta$,
    subject to the conditions that
    \begin{itemize}
        \item
        $\theta \in \mathbb{R}^{nt \times 1}$~are the private values of neighbours~$N_G(C)$, such that
        $\theta_{\nu t + i}$ is the $i \in \range{t}$th~unique private value of neighbour~$\nu \in \range{n}$ that is
        observed by any adversary in~$C$,

        \item
        $\Theta \in \mathbb{R}^{t \times 1}$~are the sums obtained by the adversaries, where $\Theta_\tau$ is the
        $\tau \in \range{t}$th~ such sum, and

        \item
        $A \in \{ 0, 1 \}^{t \times nt}$~indicates which private values are observed in which summation, such that
        $A_{\tau, \nu t + i} = 1$ if and only if the adversaries' $\tau \in \range{t}$th~summation includes the
        $i \in \range{t}$th unique private value of neighbour~$\nu \in \range{n}$.
    \end{itemize}
\end{definition}
\begin{remark}
    In \autoref{thm:reconstruction-in-multi-party-summation:the-attack:relation-to-adversarial-values}, we will show
    that it is not necessary to include adversaries' own private values in~$A\theta = \Theta$.
\end{remark}
\begin{property}
    \label{pro:reconstruction-in-multi-party-summation:obtained-adversarial-knowledge:single-private-value-per-equation}
    Let~$A$ be the adversarial knowledge over $t$~summations by~$C$.
    In each equation, each neighbour in $N_G(C)$ contributes at most one private value:
    \begin{equation}
        \label{eq:reconstruction-in-multi-party-summation:obtained-adversarial-knowledge:single-private-value-per-equation}
        \forall \tau \in \range{t}, \nu \in \range{n} :
        \sum_{i \in \range{t}} A_{\tau, \nu t + i} \in \{ 0, 1 \}.
    \end{equation}
\end{property}
\begin{property}
    \label{pro:reconstruction-in-multi-party-summation:obtained-adversarial-knowledge:sum-iff-neighbour}
    Let~$A$ be the adversarial knowledge over $t$~summations by~$C$.
    Since each equation is over all the neighbours of an adversary in~$C$, each row in~$A$ corresponds exactly
    to~$N_G(c)$ for some~$c \in C$:
    \begin{multline}
        \label{eq:reconstruction-in-multi-party-summation:obtained-adversarial-knowledge:sum-iff-neighbour}
        \forall \tau \in \range{t} : \exists c \in C : \forall \nu \in \range{n} : \\
        \left(\sum_{i \in \range{t}} A_{\tau, \nu t + i} = 1\right) \Leftrightarrow (c, N_G(C)_\nu) \in E.
    \end{multline}
    As in
    \autoref{pro:reconstruction-in-multi-party-summation:obtained-adversarial-knowledge:single-private-value-per-equation},
    the summation merely describes whether neighbour~$\nu$ is included in the $\tau$th~linear equation.
\end{property}

\subsection{Reconstruction from Adversarial Knowledge}
\label{subsec:reconstruction-in-multi-party-summation:the-attack}
Finding a (partial) solution is not trivial.
It is well-known that the reduced row echelon form~($\rref$) of a system of linear equations reveals the system's
unique solution, if it has one.
Clearly, this unique solution is also at least a partial solution.
However, if there is no unique solution, there may still be a partial solution, as in
\autoref{eq:reconstruction-in-multi-party-summation:introduction:rank-deficient-observations}.
We will show in \autoref{thm:reconstruction-in-multi-party-summation:the-attack:rref-shows-all-solutions} that finding
the reduced row echelon form of the adversarial knowledge is both necessary and sufficient to find all partial
solutions.
Moreover, we will show in
\autoref{thm:reconstruction-in-multi-party-summation:the-attack:relation-to-adversarial-values} that this is true even
if adversaries' own private values are removed from the adversarial knowledge matrix.

We begin with some definitions.
Let~$G = (V, E)$ be an undirected graph, let~$C \subseteq V$ be a set of $k$~adversaries, let~$n \coloneq \abs{N_G(C)}$,
let~$t \in \mathbb{N}$, and let~$A\theta = \Theta$ be the adversarial knowledge over $t$~summations by~$C$;
that is, $A \in \mathbb{R}^{t \times nt}$.
\begin{definition}[Solution of a variable]
    Let~$y \in \mathbb{R}^{1 \times t}$ and let~$i \in \range{nt}$.
    We say that \enquote{$y$ solves~$\theta_i$ in~$A\theta = \Theta$} if and only if the vector~$yA$ contains exactly
    one non-zero value, at index~$i$:
    \begin{equation}
        \label{eq:reconstruction-in-multi-party-summation:the-attack:y-solves-variable}
        \big( (yA)_i \neq 0 \big)
        \land \big( \forall j \in \range{nt} \setminus i : (yA)_j = 0 \big).
    \end{equation}
\end{definition}
\begin{remark}
    Since \autoref{eq:reconstruction-in-multi-party-summation:the-attack:y-solves-variable} is independent of~$\theta$
    and~$\Theta$, it is equivalent to say that \enquote{$y$ solves~$\theta_i$ in~$A$}.
\end{remark}
\begin{definition}[Partial solution]
    \label{def:reconstruction-in-multi-party-summation:the-attack:system-has-partial-solution}
    Let~$y \in \mathbb{R}^{1 \times t}$.
    If $y$ solves~$\theta_i$ in~$A$ for any~$i \in \range{nt}$, then we say that
    \enquote{$y$ is a partial solution to~$A$}.
\end{definition}

We proceed with the central theorem of this section, which states that the reduced row echelon form of~$A$ describes all
partial solutions to~$A$.
We remark that a weaker variant of this theorem is given by \citeauthor{Wang2002}~\cite{Wang2002} without a formal
proof.
\begin{theorem}
    \label{thm:reconstruction-in-multi-party-summation:the-attack:rref-shows-all-solutions}
    Let~$i \in \range{nt}$, and let~$B \in \mathbb{R}^{t \times t}$ such that~$BA = \rref(A)$.
    Then $\theta_i$ has a solution in~$A$ if and only if there exists~$r \in \range{t}$ such that~$B_r$
    solves~$\theta_i$ in~$A$.
\end{theorem}
\begin{proof}
    Given~$i \in \range{nt}$, we give a proof for both directions.

    We first prove that if there exists $r \in \range{t}$ such that $B_r$ solves~$\theta_i$, then~$\theta_i$ has a
    solution in~$A$.
    Since $A\theta = \Theta$, it follows that~$B_r A \theta = B_r \Theta$, and by
    \autoref{eq:reconstruction-in-multi-party-summation:the-attack:y-solves-variable} we have
    that~$B_r A \theta = \theta_i$.
    Therefore, $\theta_i = B_r \Theta$.
    This proves the first direction of
    \autoref{thm:reconstruction-in-multi-party-summation:the-attack:rref-shows-all-solutions}.

    We prove the other direction of
    \autoref{thm:reconstruction-in-multi-party-summation:the-attack:rref-shows-all-solutions} by contradiction.
    Let~$y \in \mathbb{R}^{1 \times t}$ be a solution to~$\theta_i$ in~$A$, so~$yA$ has its only non-zero value
    at~$(yA)_i$.
    For the sake of contradiction, assume that there is no row in~$B$ that solves~$\theta_i$ in~$A$.
    Because~$y$ is in the row space of~$A$, and the row space of~$A$ is the same as the row space of~$\rref(A)$, there
    exists~$y' \in \mathbb{R}^{1 \times t}$ such that~$yA = y' \cdot \rref(A) = y' BA$.
    By associativity of matrix multiplication, $y'B$ solves~$\theta_i$ in~$A$.
    Furthermore, since we assumed (for the sake of contradiction) that no single row of~$B$ solves~$\theta_i$ in~$A$, it
    follows that $y'$ must have multiple non-zero coefficients.
    Thus, let~$y'_r$ and~$y'_s$ be any two non-zero coefficients in~$y'$, and let~$j, k$ such that~$(BA)_{r, j}$
    and~$(BA)_{s, k}$ are the leading coefficients of their respective rows;
    these are their columns' only non-zero values, and~$j \neq k$.
    Therefore, $(yA)_j = (y' BA)_j = y'_r \neq 0$, and similarly~$(yA)_k = y'_s \neq 0$.
    However, this is a contradiction, because we initially assumed that~$yA$ has its only non-zero value at~$(yA)_i$.
    Therefore, there exists a row in~$B$ that solves~$\theta_i$ in~$A$.
    This proves the other direction of
    \autoref{thm:reconstruction-in-multi-party-summation:the-attack:rref-shows-all-solutions}.

    Therefore, it is both necessary and sufficient to check the rows of~$BA = \rref(A)$ to learn all partial solutions
    to~$A$.
\end{proof}

Note that~$A$ does not describe that adversaries know each other's private values, since~$N_G(C)$ excludes adversaries
themselves.
We show that including this knowledge does not reveal new partial solutions.
Specifically, observe that the adversarial knowledge \textit{including self-knowledge} over $t$~summations by
$k$~adversaries~$C$ is
\begin{equation}
    \label{eq:reconstruction-in-multi-party-summation:the-attack:matrix-with-self-knowledge}
    A' =
    \begin{bmatrix}
        A & R      \\
        0 & I_{tk}
    \end{bmatrix},
\end{equation}
where~$I_{tk}$ is the $(tk \times tk)$~identity matrix, $0$~is an appropriately-sized matrix of 0s, and $R$~is some
appropriately-sized binary matrix.
The rows of~$I_{tk}$ represent that adversaries know each other's values, and $R$~represents the edges between
adversaries.
\begin{theorem}
    \label{thm:reconstruction-in-multi-party-summation:the-attack:relation-to-adversarial-values}
    Let~$i < tn$.
    Then $\theta_i$ has a solution in~$A$ if and only if $\theta_i$ has a solution in~$A'$.
\end{theorem}
\begin{proof}
    Observe that
    \begin{equation}
        \label{eq:reconstruction-in-multi-party-summation:the-attack:self-knowledge-rref}
        \rref(A') =
        \begin{bmatrix}
            \rref(A) & 0      \\
            0        & I_{tk}
        \end{bmatrix},
    \end{equation}
    ignoring row-switching transformations.
    The bottom $tk$~rows solve exactly~$\theta_i$ in~$A$ for~$i \geq tn$.
    The upper rows solve~$\theta_i$ in~$A$ for~$i < tn$ if and only if the rows of~$\rref(A)$ do so.
\end{proof}
Intuitively, \autoref{thm:reconstruction-in-multi-party-summation:the-attack:relation-to-adversarial-values} holds
because the linear dependencies that exist within~$A$ remain unaffected by~$R$.

\subsection{Reconstruction Attack Feasibility}
\label{subsec:reconstruction-in-multi-party-summation:feasibility}
We show that reconstruction is feasible for honest-but-curious adversaries.
We run the attack in static graphs with randomly-placed adversaries passively collecting data.
We measure both the success rate and the number of rounds until success.
Our source code is publicly available~\cite{Dekker2023}.

\begin{remark}
    This section pertains only to static graphs.
    We show a reduction from edge-dynamic graphs to static graphs in \autoref{subsec:resistance-by-girth:dynamic-edges}.
\end{remark}

\subsubsection{Experimental setup}
By \autoref{thm:reconstruction-in-multi-party-summation:the-attack:rref-shows-all-solutions}, the success rate of the
attack depends only on the adversaries' direct neighbourhood.
Therefore, instead of modeling large peer-to-peer networks, it suffices to model only the subgraph that is relevant for
the attack.
Additionally, by \autoref{thm:reconstruction-in-multi-party-summation:the-attack:relation-to-adversarial-values}, edges
between adversaries can be ignored.
Therefore, given any graph~$G = (V, E)$ and a set of colluding adversaries~$C \subseteq V$, it suffices to model the
induced subgraph~$G[C]$, minus edges between adversaries.
This forms a bipartite graph~$H$.
We provide an example of graph induction in \autoref{fig:reconstruction-in-multi-party-summation:feasibility:graphs}.
\begin{figure}[htb]
    \centering
    \input{fig/reconstruction-feasibility-relevant-graph.tikz}
    \caption{
        A graph~$G$.
        Adversaries~$C = \{ V_1, V_2, V_3 \}$ are shaded.
        The bipartite subgraph~$H = G[C]$ consists of exactly the non-dotted nodes and edges.
    }
    \Description{
        A graph with nine nodes.
        Three nodes are adversaries and six nodes are non-adversaries.
        Of the six non-adversaries, the two nodes that are not connected to an adversary are dotted.
        All edges between adversaries and all edges between non-adversaries are dotted.
        The only elements that are not dotted are the adversaries, the direct neighbours of the adversaries, and the
        edges that connect an adversary to a non-adversary.
    }
    \label{fig:reconstruction-in-multi-party-summation:feasibility:graphs}
\end{figure}
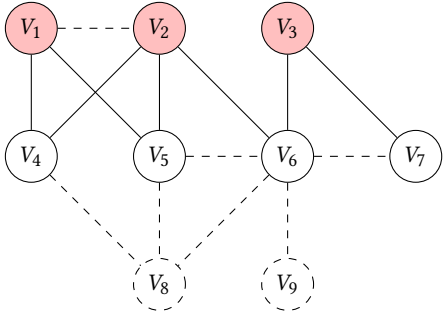

We emphasise that reconstruction depends only on the adversaries' view, regardless of the remaining graph outside this
view.
However, the likelihood of obtaining any specific adversarial view \textit{does} depend on the full graph.
For example, the probability that a random adversarial view contains a cycle depends on the connectivity of the full
graph.
For our experiments, we choose not to make assumptions on the graph's topology, analysing all possible adversarial views
equally, so that our results are agnostic to the specific network, application, and adversary.

Bipartite graphs can be parameterised by three variables:
the number of adversaries, the number of direct neighbours, and the number of edges.
We generate random graphs according to these parameter, subject to some filtering:
\begin{itemize}
    \item
    We exclude graphs in which there is a adversary with only one edge because this would allow trivial attacks, as
    described in \autoref{subsec:preliminaries:assumptions-and-notation}.

    \item
    We do \textit{not} exclude graphs in which there is an honest-but-curious user with only one edge, because this user
    may have more edges in~$G$ that are not in~$H$.

    \item
    We exclude graphs in which an honest-but-curious user has no neighbours, because these cases do not accurately
    represent the bipartite graph's parameters.

    \item
    We do \textit{not} exclude graphs in which an adversary has no neighbours.

    \item
    We do \textit{not} exclude disconnected graphs.
\end{itemize}

\subsubsection{Amount of reconstructed data}
For our first experiment, we measure the amount of private data that adversaries can reconstruct.
We generate a large amount of random bipartite graphs as described above, and count the number of partial solutions in
the biadjacency matrices.
This corresponds to the adversarial knowledge if neighbours do not update their values, and thus represents the
strongest reconstruction attack that adversaries can perform.
In \autoref{subsubsec:reconstruction-in-multi-party-summation:feasibility:rounds-before-success} we also consider
neighbours updating their values.

Firstly, we look at the proportion of data that can be reconstructed, shown in
\autoref{fig:reconstruction-in-multi-party-summation:feasibility:data-reconstructed}.
We see that if the number of adversaries is close to the number of neighbours, the adversary is typically able to
reconstruct all neighbours' data.
As the number of neighbours increases, fewer data can be reconstructed, unless compensated for by a higher connectivity.
If the graph has many neighbours and few edges, adversaries share fewer neighbours, and are thus typically unable to
exploit the overlaps in their aggregates.

Secondly, we look at the distribution of how much data can be reconstructed, shown in
\autoref{fig:reconstruction-in-multi-party-summation:feasibility:data-reconstructed-frequency}.
We see again that adversaries are more successful if they outnumber their neighbours.
As the number of neighbours increases, so does the probability of being unable to reconstruct any data.
However, even if three adversaries passively observe 15~neighbours, they still have an 11.0\%~probability of
reconstructing at least one neighbour's datum, which is unacceptable for any privacy-preserving scheme.

\subsubsection{Rounds until first reconstruction}
\label{subsubsec:reconstruction-in-multi-party-summation:feasibility:rounds-before-success}
Some partial solutions are harder to obtain than others.
For example, if the graph is such that users update their values faster than adversaries can collect them, adversaries
may never \enquote{converge} to a (partial) solution.

In the next experiment, we measure how many rounds adversaries need before reconstruction succeeds.
For each of the subgraphs in \autoref{fig:reconstruction-in-multi-party-summation:feasibility:data-reconstructed} that
were found to be susceptible to the attack, we simulate a multi-party summation protocol as follows.
Each round, a uniformly random user in the subgraph wakes up.
If the user is an adversary, they learn the sum of their neighbours' values, and adds this to the adversarial knowledge.
Otherwise, if a non-adversary wakes up, we simulate an update:
The next adversarial sum that includes this non-adversary will use a new column in the adversarial knowledge matrix.
After every round, the adversaries check for a partial solution.
We repeat this procedure 100~times to control for the order in which users wake up, truncate instances that have no
partial solutions after 250~rounds, and take the mean number of rounds until the first partial solution is found.

We show the mean number of rounds until the reconstruction attack succeeds in
\autoref{fig:reconstruction-in-multi-party-summation:feasibility:collusion-rounds}.
We see that the attack is fastest when there are more adversaries, more edges, and fewer neighbours.
Intuitively, this means that the required number of summations increases if neighbours can update their values at a
higher rate than adversaries can observe them.
For example, 3~adversaries against 15~neighbours require on average 8.8~rounds before they can reconstruct private data.
In related works such as~\cite{Vanhaesebrouck2017, Cheng2018, Cyffers2024}, users run hundreds or thousands of rounds
before the protocol terminates, significantly more than required in our attack.

\subsubsection{Conclusion of results}
We sampled all possible views of randomly selected adversaries in random graphs, excluding some trivial attack cases.
If the reconstruction attack succeeds, the adversaries obtain other users' private inputs to the
information-theoretically secure summation operation.
Our results show that passive honest-but-curious adversaries are able to obtain private data in this scenario with
non-negligible probability.
While we note that different classes of graph topologies may have varying susceptibility to reconstruction attacks, we
conclude that, in general, individually protecting each summation is insufficient for confidentiality.

\begin{figure*}[h]
    \centering

    \begin{subfigure}[t]{.33\linewidth}
        \includegraphics[width=\linewidth]{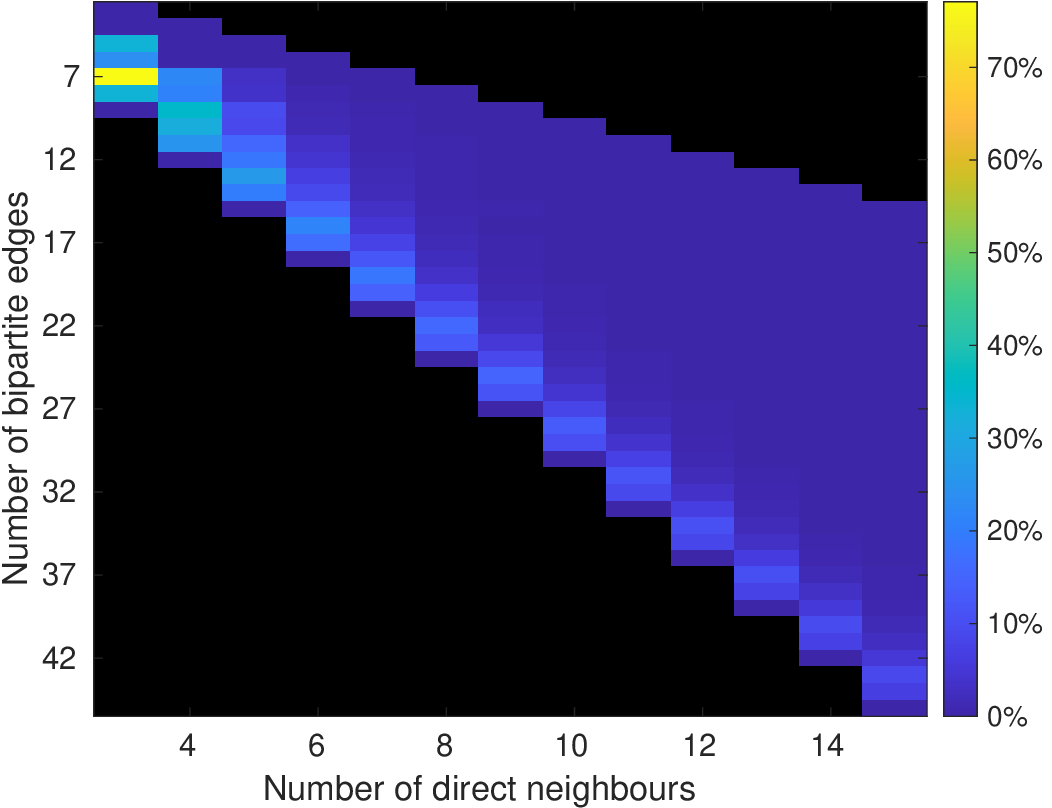}
        \caption{Three adversaries}
        \Description{
            A heatmap.
            The x-axis shows the number of direct neighbours, ranging from three to fifteen.
            The y-axis shows the number of bipartite edges, ranging from two to fourty-five.
            For each coordinate, a colour gradient shows the percentage of successful attacks.

            Some areas of the heatmap are not coloured because they indicate impossible graph configurations.
            For example, given three direct neighbours there cannot be thirty edges.
            Given the number of direct neighbours as x, the lower bound on the number of edges is x, and the upper bound
            is three times x, because this figure shows results for three adversaries.

            The figure shows that most viable graphs have a success percentage less than ten percent.
            However, as the number of edges increases, so does the success rate.
            The increase in percentage is more extreme for graphs with fewer direct neighbours.
            For fifteen direct neighbours and many edges, the success rate increases to fifteen percent.
            For three direct neighbours and many edges, the success rate increases to seventy percent.
            However, in all cases, having a complete graph drops the success rate to zero.
        }
        \label{fig:reconstruction-in-multi-party-summation:feasibility:data-reconstructed:3}
    \end{subfigure}%
    \hfil%
    \begin{subfigure}[t]{.33\linewidth}
        \includegraphics[width=\linewidth]{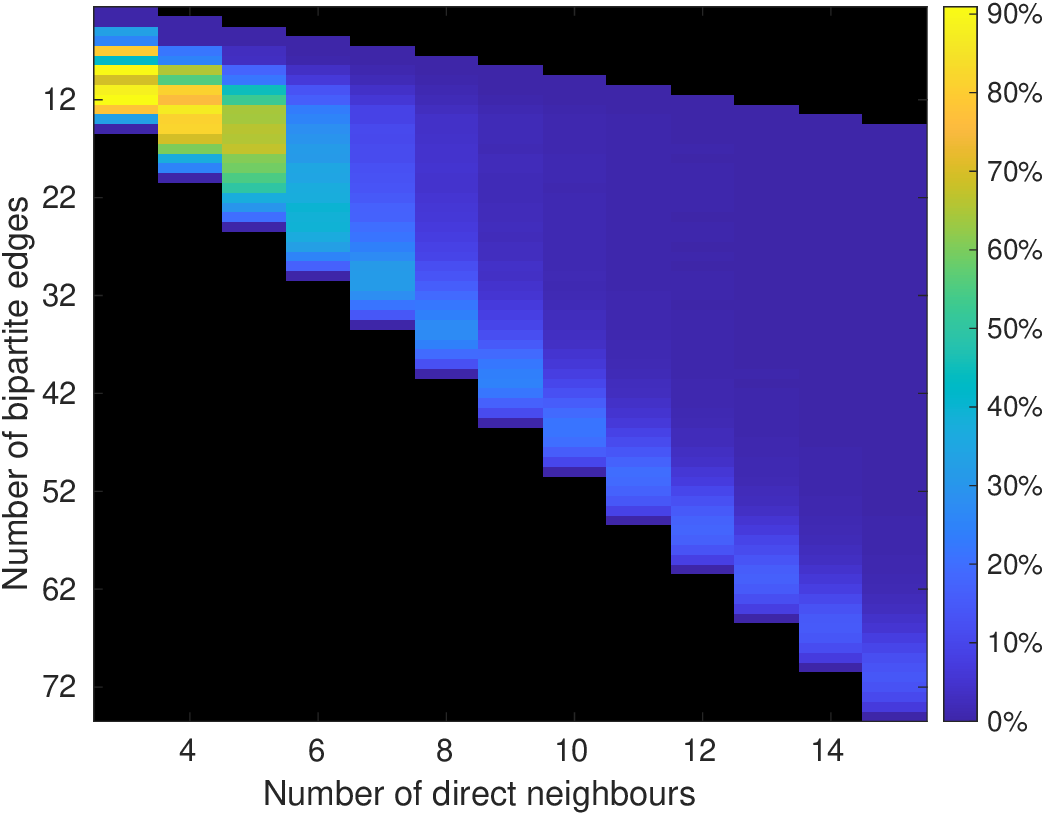}
        \caption{Five adversaries}
        \Description{
            A heatmap similar to the previous figure, Figure 4a.
            The x-axis and y-axis show similar data, but the y-axis now ranges from three to seventy-five.
            Again, many viable graphs have a success percentage of ten percent, but this increases with the number of
            edges.
            However, the success rate increases faster than in the previous figure, and there is now a larger area with
            success rates of sixty to ninety percent.
            This area is concentrated around three to five direct neighbours, starting at half the possible number of
            edges for the given number of direct neighbours, and ending a few edges above the maximum.
            As before, complete graphs have success rates of zero.
        }
        \label{fig:reconstruction-in-multi-party-summation:feasibility:data-reconstructed:5}
    \end{subfigure}%
    \hfil%
    \begin{subfigure}[t]{.33\linewidth}
        \includegraphics[width=\linewidth]{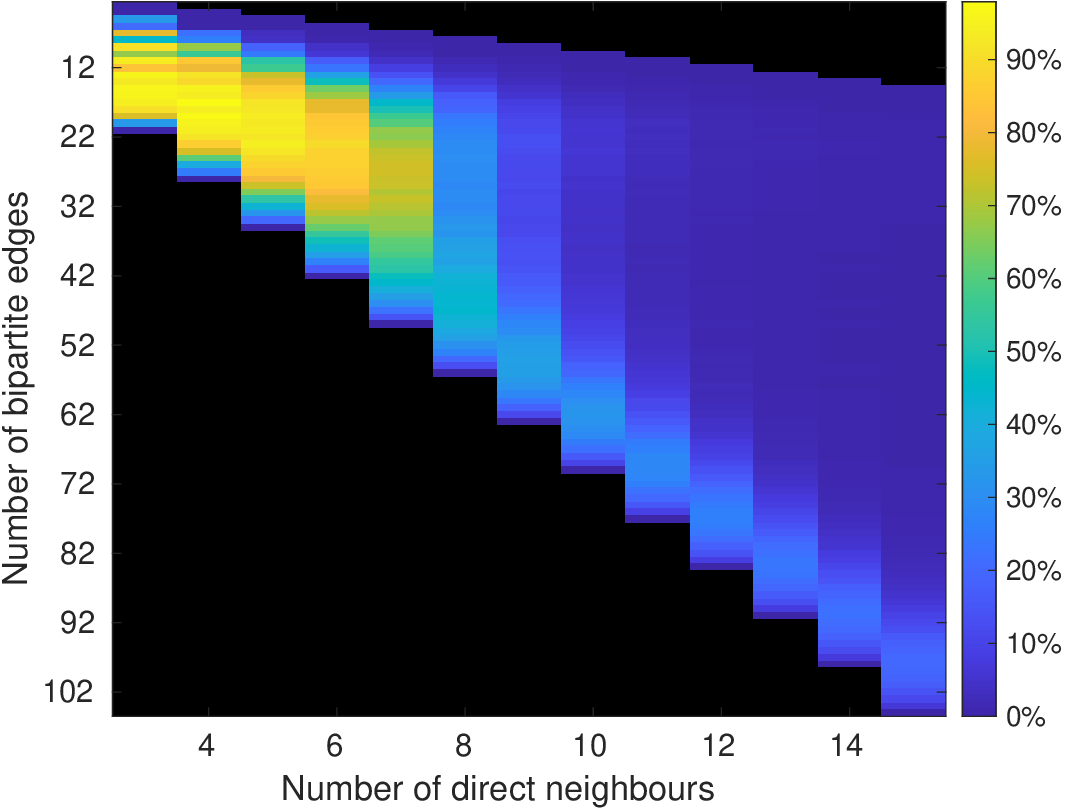}
        \caption{Seven adversaries}
        \Description{
            A heatmap similar to the previous figure, Figure 4b.
            The x-axis and y-axis show similar data, but the y-axis now ranges from three to one-hundred-and-five.
            The area of success rates higher than sixty percent is now even larger, ranging from three direct neighbours
            to seven direct neighbours, and from a third of possible edges to some portion less than the complete graph.
            This area radiates a cone towards the higher numbers of direct neighbours with success rates of thirty to
            fourty percent for eight to ten direct neighbours, decreasing as the number of direct neighbours increases.
            As before, complete graphs have success rates of zero.
        }
        \label{fig:reconstruction-in-multi-party-summation:feasibility:data-reconstructed:7}
    \end{subfigure}

    \caption{
        Proportion of neighbours' private data that can be reconstructed by adversaries.
        Each point represents the mean over 1000~random bipartite graphs.
        Black points indicate no valid bipartite graphs could be found.
        Note the different y-axes.
    }
    \label{fig:reconstruction-in-multi-party-summation:feasibility:data-reconstructed}
    \vspace{2em}

    \begin{subfigure}[t]{.33\linewidth}
        \includegraphics[width=\linewidth]{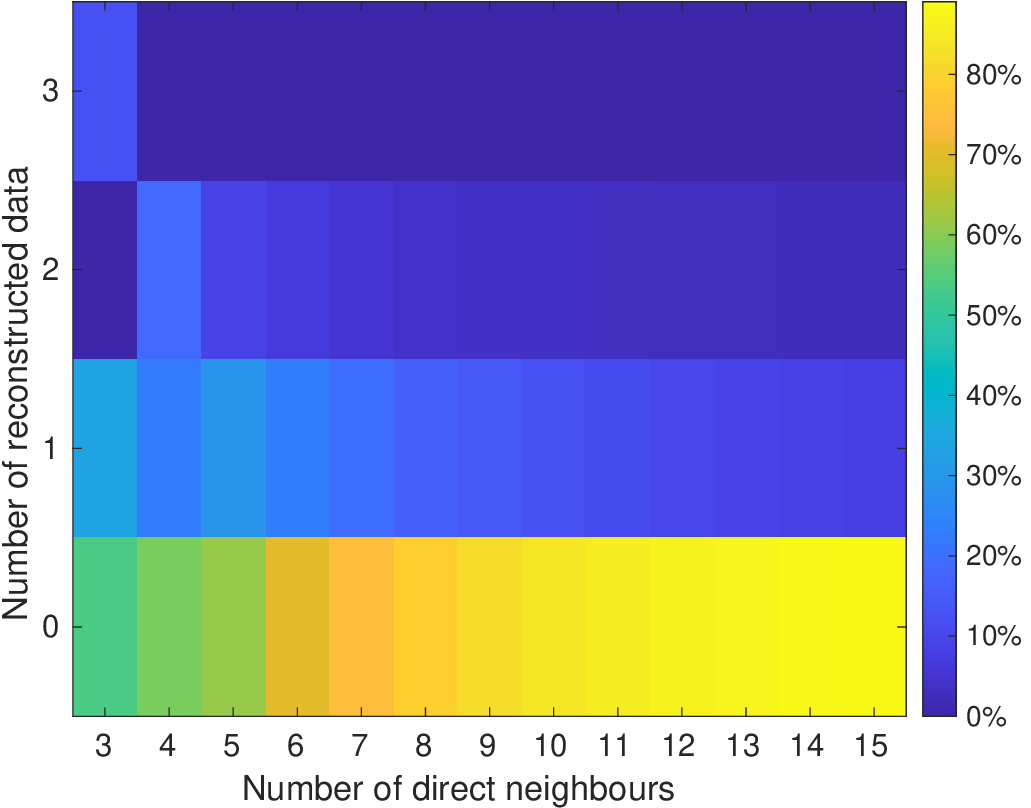}
        \caption{Three adversaries}
        \Description{
            A heatmap.
            The x-axis shows the number of direct neighbours, ranging from three to fifteen.
            The y-axis shows the number of reconstructed data, ranging from zero to three.
            For each coordinate, a colour gradient shows the relative percentage that that number of data is
            reconstructed.

            For three direct neighbours, there is, respectively for zero, one, two, or three data reconstructed, a
            relative probability of fifty, thirty, zero, or twenty percent.
            Meanwhile, for the maximum of fifteen direct neighbours, the relative probabilities are eighty, fifteen, and
            five percent.
            In between, as the number of direct neighbours increases, the probability for high amounts of reconstructed
            data drop off quickly.
        }
        \label{fig:reconstruction-in-multi-party-summation:feasibility:data-reconstructed-frequency:3}
    \end{subfigure}%
    \hfil%
    \begin{subfigure}[t]{.33\linewidth}
        \includegraphics[width=\linewidth]{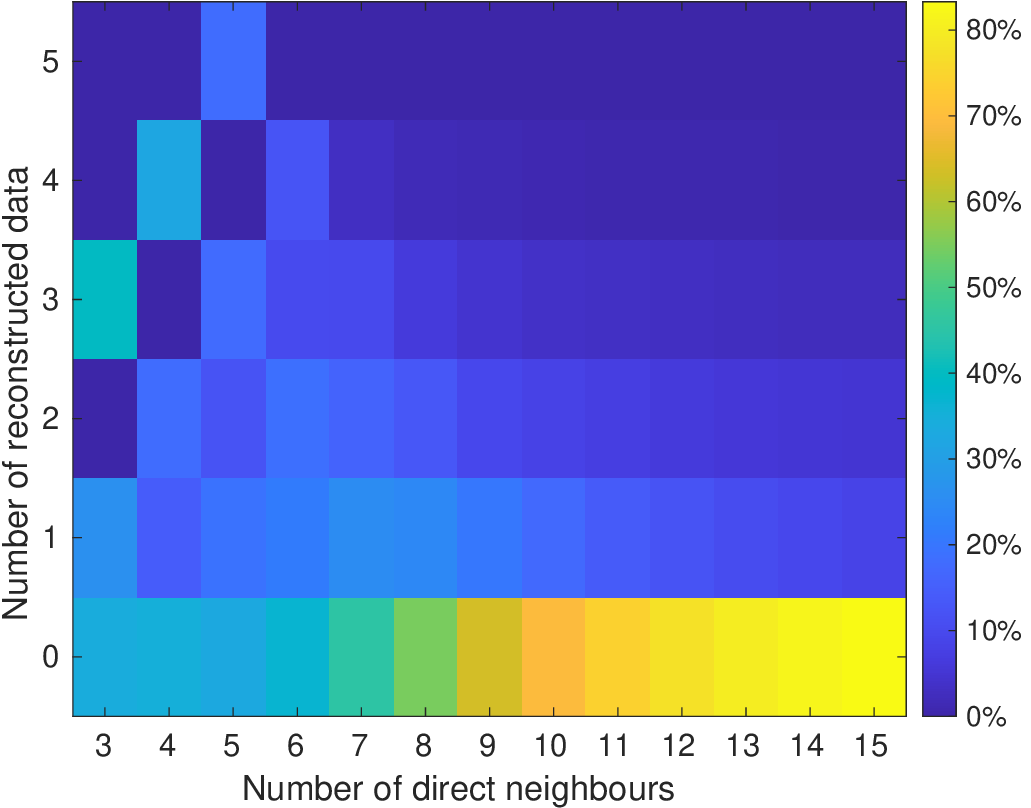}
        \caption{Five adversaries}
        \Description{
            A heatmap similar to the previous figure, Figure 5a.
            The x-axis and y-axis show similar data, but the y-axis now ranges from zero to five.
            Unlike the previous figure, three direct neighbours now has a higher probability for reconstructing more
            data, though it is obviously not possible to reconstruct more than three pieces of data.
            As before, reconstructing one datum fewer than the maximum has a zero percent chance.
            The gradient increases similar to the previous figure, but it is notable that up until five direct
            neighbours, there is always a zero percent chance of reconstructing one datum fewer than the maximum.
        }
        \label{fig:reconstruction-in-multi-party-summation:feasibility:data-reconstructed-frequency:5}
    \end{subfigure}%
    \hfil%
    \begin{subfigure}[t]{.33\linewidth}
        \includegraphics[width=\linewidth]{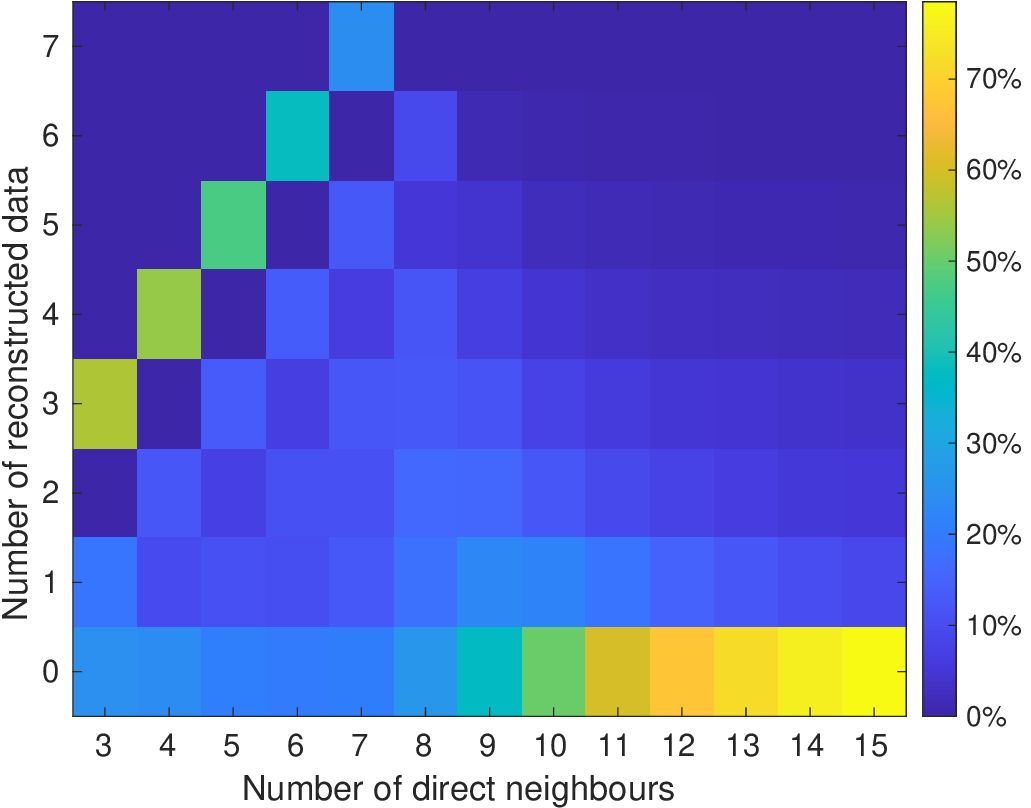}
        \caption{Seven adversaries}
        \Description{
            A heatmap similar to the previous figure, Figure 5b.
            The x-axis and y-axis show similar data, but the y-axis now ranges from zero to seven.
            The patterns observed in the previous two figures extrapolate trivially.
        }
        \label{fig:reconstruction-in-multi-party-summation:feasibility:data-reconstructed-frequency:7}
    \end{subfigure}

    \caption{
        Probability of reconstructing a given number of neighbours' data, ignoring the number of edges.
        Each column adds up to~100\%, and corresponds to a column in
        \autoref{fig:reconstruction-in-multi-party-summation:feasibility:data-reconstructed}.
    }
    \label{fig:reconstruction-in-multi-party-summation:feasibility:data-reconstructed-frequency}
    \vspace{2em}

    \begin{subfigure}[t]{.33\linewidth}
        \includegraphics[width=\linewidth]{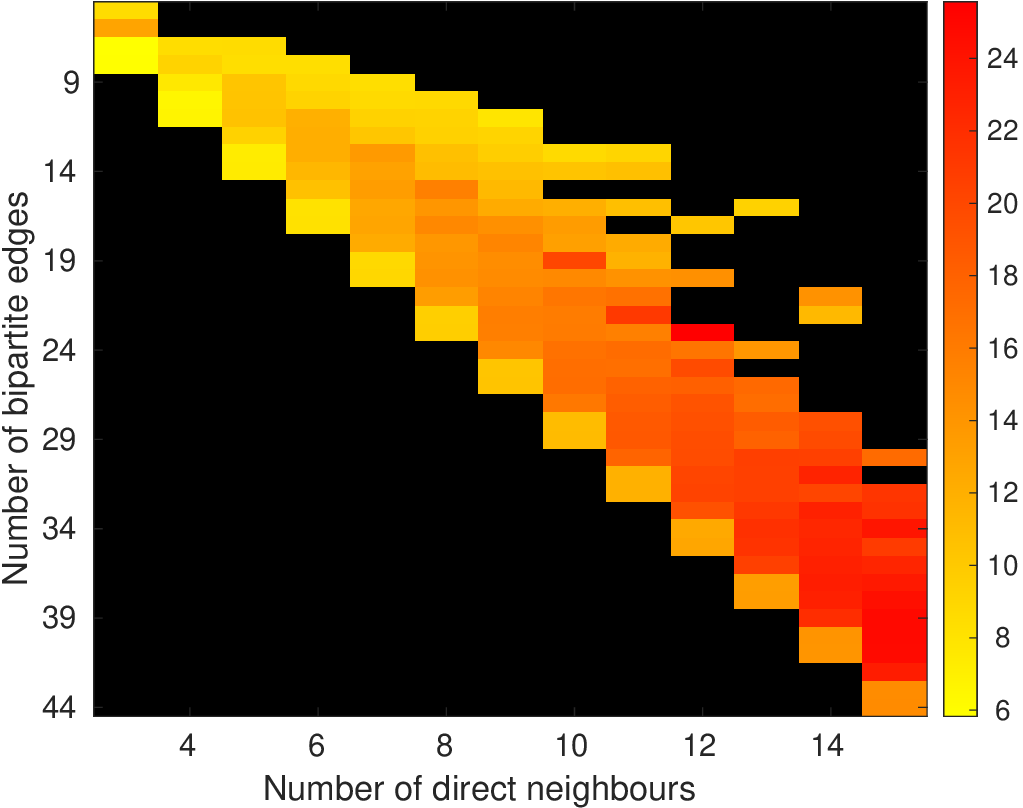}
        \caption{Three adversaries}
        \Description{
            A heatmap.
            The x-axis shows the number of direct neighbours, ranging from three to fifteen.
            The y-axis shows the number of bipartite edges, ranging from two to fourty-five.
            For each coordinate, a colour gradient shows the number of rounds needed.

            Some areas of the heatmap are not coloured because there no attacks succeeded for that combination of
            parameters.
            This is either because there exists no graph with these parameters, or because no solvable graphs were found
            with the same coordinate in Figure 4a.

            The figure shows that the number of required rounds increases as the number of direct neighbours increases.
            The number of required rounds is distinctly lower for nearly-complete graphs.
            With many direct neighbours, attacks typically did not succeed at all, showing a number of gaps in these
            areas.
        }
        \label{fig:reconstruction-in-multi-party-summation:feasibility:collusion-rounds:3}
    \end{subfigure}%
    \hfil%
    \begin{subfigure}[t]{.33\linewidth}
        \includegraphics[width=\linewidth]{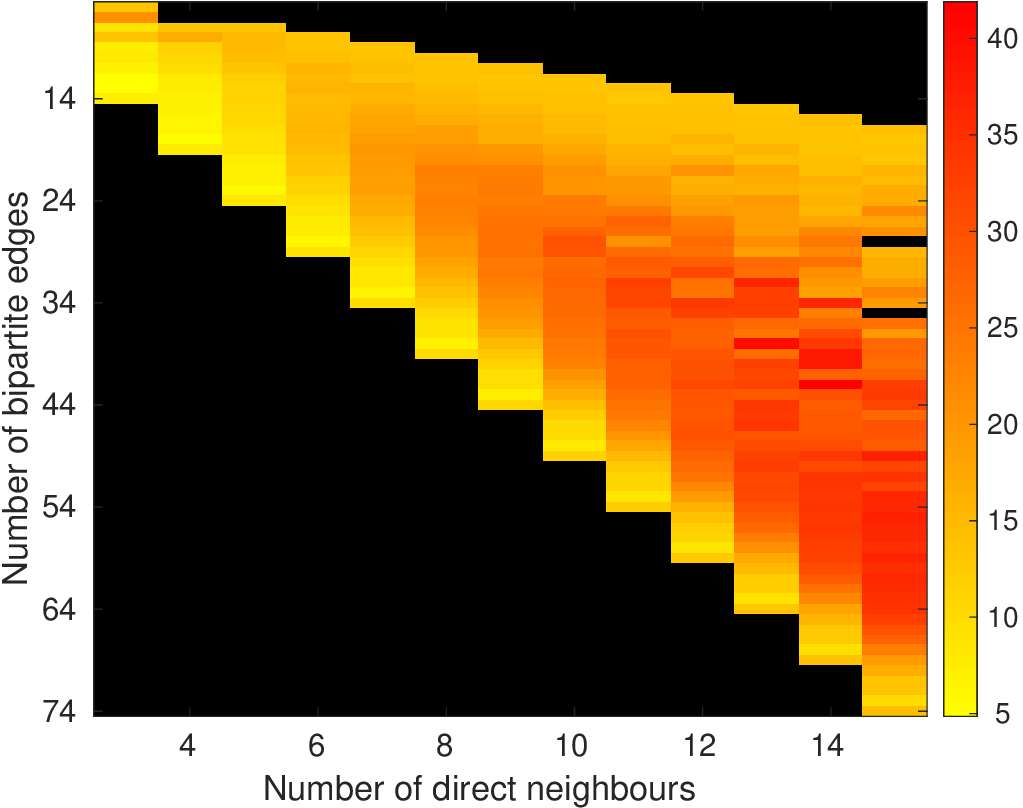}
        \caption{Five adversaries}
        \Description{
            A heatmap similar to the previous figure, Figure 6a.
            The x-axis and y-axis show similar data, but the y-axis now ranges from three to seventy-five.
            The patterns observed in the previous figure still hold, but are more exacerbated.
            Additionally, the number of holes has decreased, and in places where this figure overlaps with the previous
            figure, the number of required rounds has decreased.
        }
        \label{fig:reconstruction-in-multi-party-summation:feasibility:collusion-rounds:5}
    \end{subfigure}%
    \hfil%
    \begin{subfigure}[t]{.33\linewidth}
        \includegraphics[width=\linewidth]{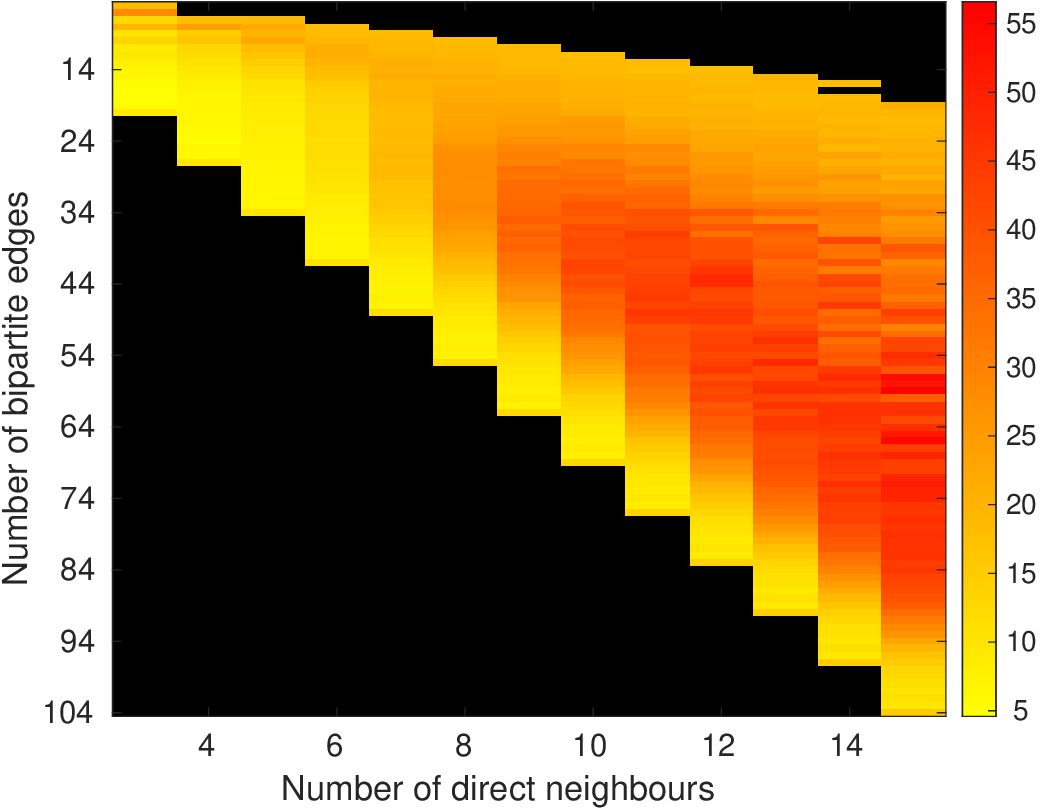}
        \caption{Seven adversaries}
        \Description{
            A heatmap similar to the previous figure, Figure 6b.
            The x-axis and y-axis show similar data, but the y-axis now ranges from three to one-hundred-and-five.
            The patterns observed in the previous two figures extrapolate trivially.
        }
        \label{fig:reconstruction-in-multi-party-summation:feasibility:collusion-rounds:7}
    \end{subfigure}

    \caption{
        Mean number of adversarial summations needed to obtain private data.
        Each point corresponds to 100~attacks on each of the solvable graphs from
        \autoref{fig:reconstruction-in-multi-party-summation:feasibility:data-reconstructed}.
    }
    \label{fig:reconstruction-in-multi-party-summation:feasibility:collusion-rounds}
\end{figure*}
\cleardoublepage

%% file: fig/adversarial-knowledge-example.tikz
\begin{tikzpicture}
[node/.style={draw, circle}, adversary/.style={fill=pink}]
    \node[node, adversary             ] (C1) {$C_1$};
    \node[node, adversary, right=of C1] (C2) {$C_2$};
    \node[node, adversary, right=of C2] (C3) {$C_3$};
    \node[node,            below=of C1] (N1) {$N_1$};
    \node[node,            right=of N1] (N2) {$N_2$};
    \node[node,            right=of N2] (N3) {$N_3$};

    \draw (C1) -- (N1);
    \draw (C1) -- (N2);
    \draw (C2) -- (N1);
    \draw (C2) -- (N3);
    \draw (C3) -- (N2);
    \draw (C3) -- (N3);
\end{tikzpicture}

%% file: fig/reconstruction-feasibility-relevant-graph.tikz
\begin{tikzpicture}
[node/.style={draw, circle}, adversary/.style={fill=pink}]
    \node[node, adversary]             (1) {$V_1$};
    \node[node, adversary, right=of 1] (2) {$V_2$};
    \node[node, adversary, right=of 2] (3) {$V_3$};
    \node[node,            below=of 1] (4) {$V_4$};
    \node[node,            right=of 4] (5) {$V_5$};
    \node[node,            right=of 5] (6) {$V_6$};
    \node[node,            right=of 6] (7) {$V_7$};
    \node[node,            below=of 5] (8) {$V_8$} [dashed];
    \node[node,            right=of 8] (9) {$V_9$} [dashed];

    \draw (1) -- (2) [dashed];
    \draw (1) -- (4);
    \draw (1) -- (5);
    \draw (2) -- (4);
    \draw (2) -- (5);
    \draw (2) -- (6);
    \draw (3) -- (6);
    \draw (3) -- (7);
    \draw (4) -- (8) [dashed];
    \draw (5) -- (6) [dashed];
    \draw (5) -- (8) [dashed];
    \draw (6) -- (7) [dashed];
    \draw (6) -- (8) [dashed];
    \draw (6) -- (9) [dashed];
\end{tikzpicture}

%% file: 05-resistence-by-girth.tex
\section{Girth as a Peer-to-Peer Reconstruction Countermeasure}\label{sec:resistance-by-girth}
In a centralised protocol, the single aggregator can track which summations have occurred, and refuse a subsequent
summation if it would result in a partial solution.
However, in a distributed computation, there is no such aggregator, and simulating the aggregator using a multi-party
protocol is impractical as this would require involving all users in each summation.
In this section, we show that to prevent reconstruction it is sufficient to increase the network's girth, which is the
length of the network's shortest cycle.
The network's girth is an established metric for peer-to-peer networks, with various peer-to-peer algorithms for
measuring and increasing the girth~\cite{Censor2021, Dolev2008, Lazebnik1995, Oliva2018}.
Using such an algorithm before running a privacy-preserving dynamic-data multi-party summation protocol is thus
sufficient to prevent reconstruction of private data by honest-but-curious adversaries.

We begin in \autoref{subsec:resistance-by-girth:single-adversary} by showing that reconstruction requires collusion.
In \autoref{subsec:resistance-by-girth:acyclic-graphs}, we show that reconstruction does not work in acyclic
graphs, regardless of the number adversaries.
In \autoref{subsec:resistance-by-girth:bounded-collusion}, generalise results to determine an upper
bound on the number of adversaries.
In \autoref{subsec:resistance-by-girth:dynamic-edges}, consider graphs with dynamic edges.
Finally, in \autoref{subsec:resistance-by-girth:impact-on-convergence}, we briefly evaluate the impact that increasing
girth has on distributed convergence.

\subsection{Privacy in Static Graphs without Collusion}\label{subsec:resistance-by-girth:single-adversary}
We begin by considering the special case of $k = 1$, i.e.\ a setting without collusion.
We show that, if the graph is static, the adversary cannot obtain other users' private values regardless of topology,
barring trivial attacks.

Assuming a privacy-preserving summation protocol, it is self-evident that repeating the summation over the same set of
values does not leak any private data.
However, while the set of neighbours is always the same in the static no-collusion setting, neighbours still update
their local values.
Thus, it remains to be shown that no reconstruction is possible with this kind of composition.
\begin{lemma}
    \label{lem:single-adversary:same-constant-for-neighbours}
    Given adversarial knowledge~$A \in \mathbb{R}^{t \times nt}$ of a single adversary with $n \geq 2$ fixed neighbours,
    we have for any~$y \in \mathbb{R}^{1 \times t}$
    \begin{equation}
        \label{eq:single-adversary:same-constant-for-neighbours}
        \forall \mu, \nu \in \range{n} :
        \sum_{i \in \range{t}} (yA)_{\mu t + i} = \sum_{i \in \range{t}} (yA)_{\nu t + i}.
    \end{equation}
    Here, $\sum_{i \in \range{t}} (yA)_{\nu t + i}$ is the sum of components of~$yA$ relating to neighbour~$\nu$.
    The equation states that in any linear combination~$yA$, every neighbour has the same sum of components.
\end{lemma}
\begin{proof}
    Firstly, because the adversary has fixed neighbours,
    \begin{equation}
        \label{eq:single-adversary:each-neighbour-in-each-equation}
        \forall \tau \in \range{t}, \nu \in \range{n} :
        \sum_{i \in \range{t}} A_{\tau, \nu t + i} = 1.
    \end{equation}
    In the linear combination~$yA$, the rows of~$A$ are scaled according to~$y$ and then summed together.
    Therefore, since each row includes each neighbour exactly once,
    \begin{equation}
        \label{eq:single-adversary:define-neighbour-constant}
        \forall \nu \in \range{n} :
        \sum_{i \in \range{t}} (yA)_{\nu t + i} = \sum_{\tau \in \range{t}} y_\tau.
    \end{equation}
\end{proof}
\begin{corollary}
    \label{cor:single-adversary:not-exactly-one-non-zero}
    Given adversarial knowledge~$A \in \mathbb{R}^{t \times nt}$ of a single adversary with $n \geq 2$ fixed neighbours,
    there exists no~$y \in \mathbb{R}^{1 \times t}$ such that~$yA$ has exactly one non-zero value.
    Therefore, there exist no partial solutions for~$A$.
\end{corollary}

\subsection{Privacy in Static Graphs with Unbounded Collusion}\label{subsec:resistance-by-girth:acyclic-graphs}
The special case of $k = 1$ provides some insights into the workings of the reconstruction attack, but not allowing any
collusion is not realistic, as honest-but-curious collusion in the form of secretly exchanging information is
undetectable and there are no strong incentives against it.
Therefore, we now proceed to consider the general case of $k \geq 1$.

Partial solutions are linear combinations of the rows of the adversarial knowledge such that all but one column cancels
out, as in \autoref{eq:reconstruction-in-multi-party-summation:introduction:solvable-system-observations}.
We already know from \autoref{cor:single-adversary:not-exactly-one-non-zero} that a partial solution requires multiple
adversaries.
If two rows in the adversarial knowledge from different adversaries match in multiple columns, then these adversaries
share multiple neighbours, and the graph has a cycle.
Otherwise, if no two rows from different adversaries overlap in multiple columns, then, since each equation has at least
two non-zero columns, each equation introduces new unknowns, taking the adversaries further from a partial solution.
In this case, if the adversaries are able to find a partial solution, they must have another row that cancels out the
unknowns of multiple other rows;
but this, too, introduces a cycle.
The intuition thus seems to be that partial solutions require a cyclic graph.
We now formally prove that this intuition is correct.
\begin{figure*}[hbt]
    \centering

    \begin{subfigure}[t]{.47\linewidth}
        \centering
        \begin{tikzpicture}
        [node/.style={draw, circle}, adversary/.style={fill=pink}]
            \node[node, adversary                   ] (C1) {$C_1$};
            \node[node,                  right=of C1] (N1) {$N_1$};
            \node[node, adversary, below right=of N1] (C2) {$C_2$};
            \node[node,             below left=of C2] (N2) {$N_2$};
            \node[node, adversary,        left=of N2] (C3) {$C_3$};
            \node[node,             above left=of C3] (N3) {$N_3$};
            \node[node, adversary,       right=of N1] (C4) {$C_4$};
            \node[node,                  right=of C4] (N4) {$N_4$};

            \draw (C1) -- (N1);
            \draw (N1) -- (C2);
            \draw (C2) -- (N2);
            \draw (N2) -- (C3);
            \draw (C3) -- (N3);
            \draw (N3) -- (C1);
            \draw (N1) -- (C4);
            \draw (C4) -- (N4);
        \end{tikzpicture}
        \caption{
            A graph~$G$ featuring adversaries~$C = \{ C_1, C_2, C_3, C_4 \}$ and non-adversaries
            $N = \{ N_1, N_2, N_3, N_4 \}$.
        }
        \Description{
            A simple graph featuring a cycle with one bit sticking out.
            The bit that sticks out is a short chain of nodes.
            The cycle features three adversaries, C1, C2, and C3, and three non-adversaries, N1, N2, and N3, in
            alternating order.
            The bit that sticks out of the cycle is a chain of users, consisting in order of users N1, C4, and N4.
        }
        \label{fig:resistance-by-girth:acyclic-graphs:example-graph:before}
    \end{subfigure}
    \hfil
    \begin{subfigure}[t]{.47\linewidth}
        \setcounter{subfigure}{3}
        \centering
        \begin{tikzpicture}
        [node/.style={draw, circle}, adversary/.style={fill=pink}]
            \node[node, adversary                   ] (C1) {$C_1$};
            \node[node,                  right=of C1] (N1) {$N_1$};
            \node[node, adversary, below right=of N1] (C2) {$C_2$};
            \node[node,             below left=of C2] (N2) {$N_2$};
            \node[node, adversary,        left=of N2] (C3) {$C_3$};
            \node[node,             above left=of C3] (N3) {$N_3$};
            \node[node, adversary,       right=of N1] (C4) {$C_4$};

            \draw (C1) -- (N1);
            \draw (N1) -- (C2);
            \draw (C2) -- (N2);
            \draw (N2) -- (C3);
            \draw (C3) -- (N3);
            \draw (N3) -- (C1);
        \end{tikzpicture}
        \caption{The bipartite graph~$H$ corresponding to biadjacency matrix~$A''$.}
        \Description{
            This graph is identical to graph G, except that user N4 is gone, and the edge between N1 and C4 is removed.
            The cycle thus remains intact.
        }
        \label{fig:resistance-by-girth:acyclic-graphs:example-graph:after}
    \end{subfigure}

    \vspace{9mm}  
    \begin{subfigure}{\linewidth}
        \setcounter{subfigure}{1}
        \begin{align*}
            &
            A =
            \begin{bNiceArray}{ccccc|ccccc|ccccc|ccccc}[margin]
                1 & 0 & 0 & 0 & 0  &  0 & 0 & 0 & 0 & 0  &  1 & 0 & 0 & 0 & 0  &  0 & 0 & 0 & 0 & 0 \\
                1 & 0 & 0 & 0 & 0  &  1 & 0 & 0 & 0 & 0  &  0 & 0 & 0 & 0 & 0  &  0 & 0 & 0 & 0 & 0 \\
                0 & 0 & 0 & 0 & 0  &  1 & 0 & 0 & 0 & 0  &  1 & 0 & 0 & 0 & 0  &  0 & 0 & 0 & 0 & 0 \\
                0 & 0 & 0 & 0 & 0  &  1 & 0 & 0 & 0 & 0  &  0 & 1 & 0 & 0 & 0  &  0 & 0 & 0 & 0 & 0 \\
                1 & 0 & 0 & 0 & 0  &  0 & 0 & 0 & 0 & 0  &  0 & 0 & 0 & 0 & 0  &  1 & 0 & 0 & 0 & 0
                \CodeAfter
                \OverBrace[shorten, yshift=3pt]{1-1}{5-5}{N_1}
                \OverBrace[shorten, yshift=3pt]{1-6}{5-10}{N_2}
                \OverBrace[shorten, yshift=3pt]{1-11}{5-15}{N_3}
                \OverBrace[shorten, yshift=3pt]{1-16}{5-20}{N_4}
            \end{bNiceArray},
            &&
            A' =
            \begin{bNiceArray}{cccc}[margin]
                1 & 0 & 1 & 0 \\
                1 & 1 & 0 & 0 \\
                0 & 1 & 1 & 0 \\
                0 & 1 & 1 & 0 \\
                1 & 0 & 0 & 1
            \end{bNiceArray},
            &&
            A'' =
            \begin{bNiceArray}{cccc}[margin]
                1 & 0 & 1 & 0 \\
                1 & 1 & 0 & 0 \\
                0 & 1 & 1 & 0
            \end{bNiceArray}
        \end{align*}

        \caption{
            The adversarial knowledge~$A$ after users from
            \autoref{fig:resistance-by-girth:acyclic-graphs:example-graph:before} run in the
            sequence~$(C_1, C_2, C_3, N_3, C_3, C_4)$; the matrix~$A'$ with collapsed columns; and the matrix~$A''$
            without duplicate and unused rows.
        }
        \label{fig:resistance-by-girth:acyclic-graphs:example:knowledge}
    \end{subfigure}

    \begin{subfigure}{\linewidth}
        \begin{align*}
            y =
            \begin{bNiceArray}{ccccc}[margin]
                1 & 1 & -1 & 0 & 0
            \end{bNiceArray},
            &&
            y' =
            \begin{bNiceArray}{ccccc}[margin]
                1 & 1 & -1 & 0 & 0
            \end{bNiceArray},
            &&
            y'' =
            \begin{bNiceArray}{ccc}[margin]
                1 & 1 & -1
            \end{bNiceArray}
        \end{align*}

        \caption{Partial solutions respectively of~$A$, $A'$, and~$A''$.}
        \label{fig:resistance-by-girth:acyclic-graphs:example:solutions}
    \end{subfigure}

    \caption{
        Example transformation of graph and adversarial knowledge as seen in the proof of
        \autoref{thm:resistance-by-girth:acyclic-graphs:acyclic-is-no-solution}.
    }
    \label{fig:resistance-by-girth:acyclic-graphs:example}
\end{figure*}
\begin{theorem}
    \label{thm:resistance-by-girth:acyclic-graphs:acyclic-is-no-solution}
    Let~$G = (V_G, E_G)$ be an undirected graph, let~$C \subseteq V_G$ be the set of adversaries,
    let~$k \coloneq \abs{C}$, let~$n \coloneq \abs{N_G(C)}$, let~$t$ be the number of summations performed by the
    adversaries~$C$, and let~$A \in \mathbb{R}^{t \times nt}$ be the adversarial knowledge.

    If~$G$ is acyclic, then~$A$ does not have partial solutions.
\end{theorem}
\begin{proof}
    We give a proof by contraposition:
    Given a partial solution to~$A$, we show that~$G$ is cyclic.
    Let~$y \in \mathbb{R}^{1 \times t}$ be a partial solution to~$A$.
    We show how to find a bipartite subgraph~$H$ of~$G$ such that its biadjacency matrix~$A''$ has a partial
    solution~$y''$.
    We then show that this implies the existence of a cycle in~$G$.
    Our proof works in multiple steps:
    \begin{enumerate*}[label=(\arabic*)]
        \item combine columns of~$A$ to create~$A'$,
        \item remove rows from~$A'$ to create~$A''$,
        \item create the corresponding partial solution~$y''$, and finally
        \item show that $G$~is cyclic.
    \end{enumerate*}
    We show an example of this procedure in \autoref{fig:resistance-by-girth:acyclic-graphs:example}.

    \begin{enumerate}
        \item
        \textit{Combine columns.}
        We merge the $t$~columns in~$A$ assigned to each neighbour to obtain~$A'$.
        Let~$y' = y$, and let~$A' \in \mathbb{R}^{t \times n}$ such that
        \begin{equation}
            \label{eq:resistance-by-girth:acyclic-graphs:merge-columns}
            \forall \tau \in \range{t}, \nu \in \range{n} :
            A'_{\tau, \nu} \coloneq \sum_{i \in \range{t}} A_{\tau, \nu t + i}.
        \end{equation}
        It follows from
        \autoref{pro:reconstruction-in-multi-party-summation:obtained-adversarial-knowledge:single-private-value-per-equation}
        that this is a binary matrix, and it follows from
        \autoref{pro:reconstruction-in-multi-party-summation:obtained-adversarial-knowledge:sum-iff-neighbour} that no
        neighbour relations are removed.
        Furthermore, observe that
        \begin{equation}
            \label{eq:resistance-by-girth:acyclic-graphs:merge-columns-keeps-solution}
            \forall \nu \in \range{n} : (y'A')_\nu = \sum_{i \in \range{t}} (yA)_{\nu t + i}.
        \end{equation}
        Since~$yA$ contains exactly one non-zero value, so does~$y'A'$.
        Therefore, $y'$~is a partial solution to~$A'$.

        \item
        \textit{Remove rows.}
        We remove duplicate and unused rows from~$A'$ to obtain~$A''$.
        We define~$A''$ as a set of rows:
        \begin{align}
            A'' \coloneq \{ A'_i \mid
            \ &i \in \range{t}\ \land \\
            \label{eq:resistance-by-girth:acyclic-graphs:remove-rows-a-unique}
            &\nexists j \in \range{i} : A'_i = A'_j\ \land \\
            \label{eq:resistance-by-girth:acyclic-graphs:remove-rows-a-used}
            &{\scriptstyle\sum} \{ y'_j \mid j \in \range{t} \land A'_i = A'_j \} \neq 0 \}.
        \end{align}
        Here, \autoref{eq:resistance-by-girth:acyclic-graphs:remove-rows-a-unique} excludes duplicates by only choosing
        row~$A'_i$ if there is no~$j < i$ such that~$A'_i = A'_j$, and
        \autoref{eq:resistance-by-girth:acyclic-graphs:remove-rows-a-used} excludes unused rows by only picking
        row~$A'_i$ if the sum of~$y'_j$ over all identical rows~$A'_j$ is non-zero.

        \item
        \textit{Create partial solution.}
        We similarly combine and remove the corresponding columns from~$y'$ to obtain~$y''$.
        To do so, we define a function~$\phi$ that describes how the rows of~$A''$ relate to the rows of~$A'$.
        Let~$s$ be the number of rows in~$A''$.
        Then we define~$\phi : \range{s} \rightarrow \range{t}^\ast$ such that
        \begin{equation}
            \label{eq:resistance-by-girth:acyclic-graphs:phi}
            \forall \tau \in \range{t}, \sigma \in \range{s} :
            \tau \in \phi(\sigma) \Leftrightarrow A'_\tau = A''_\sigma.
        \end{equation}
        Using this function, we define~$y'' \in \mathbb{R}^{1 \times s}$ as
        \begin{equation}
            \label{eq:resistance-by-girth:acyclic-graphs:remove-rows-y}
            \forall \sigma \in \range{s} : y''_\sigma \coloneq \sum_{\tau \in \phi(\sigma)} y'_\tau.
        \end{equation}
        It follows that
        \begin{align}
            \label{eq:resistance-by-girth:acyclic-graphs:remove-rows-keeps-solution}
            \forall \nu \in \range{n} : (y''A'')_\nu
            &= \sum_{\sigma \in \range{s}} (y''_\sigma A''_{\sigma, \nu}) \\
            &= \sum_{\sigma \in \range{s}} \sum_{\tau \in \phi(\sigma)} (y'_\tau A''_{\sigma, \nu}) \\
            &= \sum_{\sigma \in \range{s}} \sum_{\tau \in \phi(\sigma)} (y'_\tau A'_{\tau, \nu}) \\
            &= \sum_{\tau \in \range{t}} (y'_\tau A'_{\tau, \nu}) \\
            &= (y'A')_\nu.
        \end{align}
        Therefore, $y''A'' = y'A'$, and $y''$ is a partial solution to~$A''$.

        \item
        \textit{Find cycle.}
        Note that $A''$ is the biadjacency matrix of some bipartite subgraph~$H = (C', N_G(C), E_H)$ of $G$,
        where $C' \subseteq C$ and $E_H \subseteq E_G$.
        Assume, for the sake of contradiction, that $H$ is acyclic.
        Then $H$ has two distinct nodes~$i, j$ with degree one.
        Since adversaries cannot have degree one in~$G$, and
        $\forall c \in C' : \left( N_H(c) = N_G(c) \lor N_H(c) = \emptyset \right)$, we know that $i, j \in N_G(C)$.
        Consequently, the columns in~$A''$ for~$i, j$ must each contain only one non-zero value, and $y''$ does not
        contain zeroes at all by \autoref{eq:resistance-by-girth:acyclic-graphs:remove-rows-a-used}.
        Therefore, $(y'' A'')_i \neq 0$ and~$(y'' A'')_j \neq 0$.
        However, this implies that $y'' A''$ has multiple non-zero values, which contradicts the earlier observation
        that~$y''$ is a partial solution to~$A''$.
        Therefore, $H$~is cyclic, and so is~$G$.
    \end{enumerate}
\end{proof}

Our proof shows that partial solutions imply the existence of cycles.
However, this does not mean that cycles imply the existence of partial solutions.
Indeed, we show in \autoref{subsec:resistance-by-girth:bounded-collusion} that structured cycles can be introduced
without creating partial solutions.
\begin{remark}
    \autoref{thm:resistance-by-girth:acyclic-graphs:acyclic-is-no-solution} pertains only to \textit{partial} solutions.
    Even in an acyclic topology, there may be linear relations that reveal sensitive information without leaking private
    values outright, such as~$\theta_1 = \theta_2$ or~$\theta_3 = 4 \times \theta_5$.
    Protecting these relations is left for future work.
\end{remark}

\subsection{Privacy in Static Graphs with Bounded Collusion}\label{subsec:resistance-by-girth:bounded-collusion}
While acyclic graphs resist reconstruction attacks, these graphs are not well-suited for peer-to-peer networks for two
reasons.
Firstly, if any non-leaf node becomes unavailable, the network becomes disconnected.
Secondly, leaf nodes have only one neighbour, and thus cannot initiate summations to learn from their neighbours.

We show that no partial solutions exist given an upper bound on the number of adversaries.
This bound depends on the graph's girth, which is the length of its shortest cycle.
\begin{theorem}
    \label{thm:resistance-by-girth:bounded-collusion:girth-determines-bound}
    Let~$G = (V_G, E_G)$ be an undirected graph, let~$C \subseteq V_G$ be a set of $k$~adversaries,
    let~$n \coloneq \abs{N_G(C)}$, let~$t$ be the number of summations performed by~$C$, and
    let~$A \in \mathbb{R}^{t \times nt}$ be the adversarial knowledge.

    If~$\girth(G) > 2k$, then~$A$ does not have partial solutions.
\end{theorem}
\begin{proof}
    We give a proof by contraposition:
    Given a partial solution to~$A$, we show that~$\girth(G) \leq 2k$.
    Let~$H$ be as in the proof of \autoref{thm:resistance-by-girth:acyclic-graphs:acyclic-is-no-solution}.
    Then $H$~is cyclic.
    Since~$H$ is bipartite, every edge in the cycle is between an adversary and a neighbour.
    Since each node in the cycle is visited at most once, the cycle length is at most~$2k$.
    This cycle also exists in~$G$.
    Therefore, $\girth(G) \leq 2k$.
\end{proof}

\subsection{Privacy in Dynamic Graphs}\label{subsec:resistance-by-girth:dynamic-edges}
So far, we have assumed that graphs are static.
However, this prevents users from changing their neighbours, which is unrealistic if users move through the network.
We briefly show that dynamic graphs can be reduced to static graphs.

If a single user performs two summations over two sets of neighbours, they learn exactly the same information as two
users would over those same sets of neighbours.
We show an example in \autoref{fig:resistance-by-girth:dynamic-edges:example}.
\begin{figure}[htb]
    \centering

    \begin{subfigure}{.49\linewidth}
        \centering
        \input{fig/dynamic-graph-example-dynamic.tikz}
        \caption{A dynamic graph. The dotted edge is not present in all rounds.}
        \Description{
            An undirected graph with four nodes, U, N1, N2, and N3.
            Node U is connected to the three other nodes.
            The edge between U and N3 is dotted.
            The other three nodes are not connected to each other.
        }
        \label{fig:resistance-by-girth:dynamic-edges:example:dynamic}
    \end{subfigure}

    \begin{subfigure}{.49\linewidth}
        \centering
        \input{fig/dynamic-graph-example-static.tikz}
        \caption{A reduction to a static graph. $U$ has been split into~$U_1$ and~$U_2$.}
        \Description{
            An undirected graph with five nodes, U1, U2, N1, N2, and N3.
            Node U1 is connected to nodes N1 and N2.
            Node U2 is connected to nodes N1, N2, and N3.
            Nodes N1, N2, and N3 are not connected to each other.
            A dotted box is drawn around U1 and U2.
        }
        \label{fig:resistance-by-girth:dynamic-edges:example:static}
    \end{subfigure}

    \caption{
        Example of how a dynamic graph can be reduced to a static graph.
        $U$~learns the same as~$U_1$ and~$U_2$ together.
    }
    \label{fig:resistance-by-girth:dynamic-edges:example}
\end{figure}
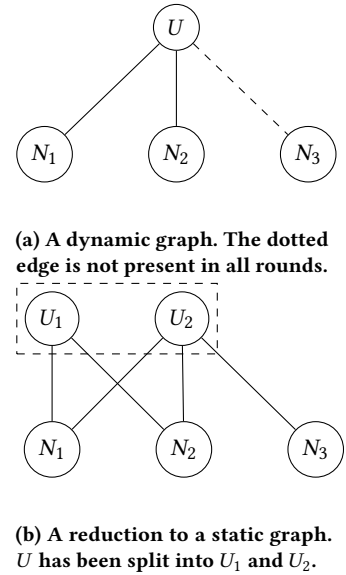
More generally, $k$~users with static neighbours can learn the exact same information as $\ell$~users with $k$~different
sets of neighbours.
Our results on reconstruction feasibility in static graphs from
\autoref{subsec:reconstruction-in-multi-party-summation:feasibility} can be translated similarly to dynamic graphs.

We conclude that \autoref{thm:resistance-by-girth:bounded-collusion:girth-determines-bound} implies the following.
\begin{corollary}
    Let~$G = (V_G, E_G)$ be a dynamic undirected graph, let~$C \subseteq V_G$ be a set of adversaries,
    let~$n \coloneq \abs{N_G(C)}$, let~$t$ be the number of summations performed by~$C$, let~$k$ be the number of sets
    of neighbours the adversaries sum over, and let~$A \in \mathbb{R}^{t \times nt}$ be the adversarial knowledge.

    If~$\girth(G) > 2k$, then~$A$ does not have partial solutions.
\end{corollary}

There are several important limitations to this result.
Firstly, the upper bound on the number of adversaries depends on the girth, but the girth may not be known beforehand if
users move through the network in unpredictable ways.
Secondly, even if a minimum girth is guaranteed throughout the protocol, the upper bound implies a maximum number of
changes that may occur during the protocol.

\subsection{Impact on Convergence}\label{subsec:resistance-by-girth:impact-on-convergence}
We briefly evaluate the impact of increasing the network's girth on the convergence of a protocol running over that
network.
Specifically, we numerically simulate a distributed averaging protocol~\cite{Boyd2004}, which is just a
non-privacy-preserving form of distributed learning.
We intentionally choose a simple, efficient, non-noisy protocol to make the impact of the girth parameter most apparent.
The \enquote{numerical simulation} part of the description is because we do not actually create separate processes and
communication for the nodes.
Our source code is publicly available~\cite{Dekker2023}.

We use the system model presented in \autoref{subsec:preliminaries:assumptions-and-notation}.
We create a network by generating a random Erdős--Rényi graph with 50~nodes and with each edge having a probability~$p$
of being added.
Each node holds a single private scalar value, sampled uniformly from the range~$\{0 \ldots 50 \}$.
Each round, one random node updates their private value to be the unweighted mean of their neighbours' values and their
own value.
We then measure the number of rounds until convergence, and take the mean over 1000~repetitions of this procedure.
We define convergence as the moment at which any two nodes' local values differ by at most 1.
Changing this threshold does not give fundamentally different results.

To measure the effect girth has on convergence, we \enquote{stretch} graphs to a given girth by iteratively removing
random edges from cycles shorter than the desired girth until no such cycles remain.
With 50~nodes, stretching to a girth of~$x$ ensures reconstruction attacks are impossible when less
than~$\sfrac{\sfrac{x}{2}}{50} = x\%$ of users collude.
For example, after stretching the girth to 25, the graph can resist collusions of less than 25\% of users.

\begin{figure}[htb]
    \centering
    \includegraphics[width=\linewidth]{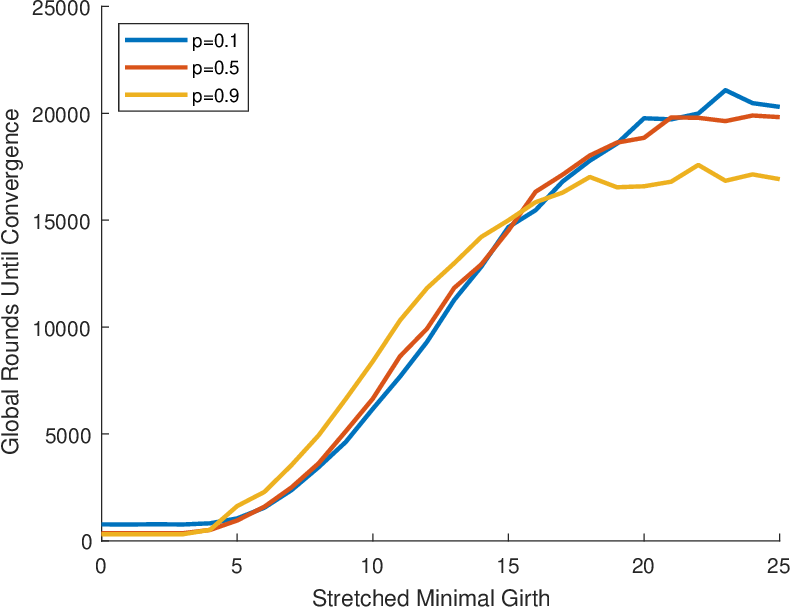}
    \caption{
        Number of rounds until convergence in distributed averaging in random Erdős--Rényi graphs with 50~nodes and
        varying edge probabilities~$p$, as a function of the girth to which the graphs are \enquote{stretched}.
    }
    \Description{
        A line chart with three lines.
        The three lines seperately represent p equals 0.1, p equals 0.5, and p equals 0.9.
        The x-axis shows the stretched minimal girth, ranging from zero to twenty-five.
        The y-axis shows the number of global rounds until convergence, ranging from zero to twenty-five-thousand.

        The line for p equals 0.1 starts at girth zero, requiring two-hundred rounds before the protocol converges.
        Up to girth five, the number of rounds does not change significantly.
        Then, between girth six and fifteen, the line increases linearly to fifteen-thousand rounds.
        Finally, between girth sixteen and twenty-five, the slope decreases, settling at twenty-one-thousand rounds.

        The line for p equals 0.5 is nearly identical, but requires slightly more rounds at the highest girths.

        The line for p equals 0.9 shows the same pattern as the previous two lines, but requires slightly more rounds at
        low girths, and requires slightly fewer rounds at high girths.
        At girth 25, this line requires only seventeen-thousand rounds, while the other two lines require
        twenty-one-thousand rounds.
    }
    \label{fig:resistance-by-girth:impact-on-convergence-erdos-renyi}
\end{figure}
We show our results in \autoref{fig:resistance-by-girth:impact-on-convergence-erdos-renyi}.
Since undirected graphs always have girth at least~3, no significant changes occur at these low girths.
As the girth increases, so does the number of rounds required.
As the girth approaches 25, the slope approaches zero.
Graphs that initially have more edges (as determined by~$p$) require more rounds at low girths, but settle at a lower
number of rounds at high girths.
When we look at our experiments in more detail, we see that ceilings occur once all cycles have been removed, and that
graphs with high~$p$ retain more edges.
This matches the intuition that information propagates more efficiently when there are more edges.

Our results show that increasing girth affects convergence speed significantly.
Though state-of-the-art distributed learning protocols typically require several tens of thousands of
rounds~\cite{Vanhaesebrouck2017, Cheng2018, Cyffers2024}, the magnitude by which increasing girth increases the number
of required rounds may be excessive for some applications.
More sophisticated edge removal methods may ameliorate this issue.
Furthermore, though implementing the cycle removal method from our experiment above as a distributed protocol is
trivial,%
\footnote{A node can break all cycles of at most length~$\ell$ that they are part of as follows.
The node floods a unique random message, paired with a counter starting at~$\ell$, through the network.
Each time a node forwards the message, the counter is decreased.
Once the counter reaches zero, nodes stop forwarding the message.
If (and only if) the source node receives back their own message, they are part of a cycle of length at most~$\ell$, and
remove the edge on which the message came in.}
these methods are not necessarily communicationally efficient.
To the best of our knowledge, there is no research on communication-efficient distributed \enquote{graph stretching}.
That said, there are distributed protocols for measuring the network's girth~\cite{Censor2021} and for removing
\emph{all} cycles~\cite{Dolev2008, Oliva2018}.
We conclude that determining a network's resistance by measuring the girth is feasible in general, but increasing girth
is practical only when communication efficiency is not a concern.

%% file: fig/dynamic-graph-example-dynamic.tikz
\begin{tikzpicture}
[node/.style={draw, circle}]
    \node[node,]             (U)  {$U$};
    \node[node, below=of U]  (N2) {$N_2$};
    \node[node, left =of N2] (N1) {$N_1$};
    \node[node, right=of N2] (N3) {$N_3$};

    \draw (U) -- (N1);
    \draw (U) -- (N2);
    \draw (U) -- (N3) [dashed];
\end{tikzpicture}

%% file: fig/dynamic-graph-example-static.tikz
\begin{tikzpicture}
[node/.style={draw, circle}]
    \node[node,]             (U1) {$U_1$};
    \node[node, right=of U1] (U2) {$U_2$};
    \node[node, below=of U1] (N1) {$N_1$};
    \node[node, right=of N1] (N2) {$N_2$};
    \node[node, right=of N2] (N3) {$N_3$};

    \draw (U1) -- (N1);
    \draw (U1) -- (N2);
    \draw (U2) -- (N1);
    \draw (U2) -- (N2);
    \draw (U2) -- (N3);

    \node[draw, dashed, fit=(U1) (U2)] {};
\end{tikzpicture}

%% file: 06-conclusion.tex
\section{Conclusion}\label{sec:conclusion}
We investigated reconstruction attacks in the setting of secure multi-party computation.
We observed that existing multi-party computation literature does not consider protocols in which intermediate values
are intentionally exposed by the ideal functionality, and seemingly assumes that protocols are not self-composed when
deployed.
In our investigation, we focused on a peer-to-peer setting with privacy-preserving summation in which users' data
change over time.
In random subgraphs with 18~users, we found that three passive honest-but-curious adversarial users have an
11.0\%~success rate at recovering another user's private data using a reconstruction attack, requiring an average of
8.8~rounds per adversary.
We analysed the structural dependencies of the underlying network graph that permit this attack, and proved that
successful reconstruction attacks correspond to cycles in the network.
More generally, we showed that the length of the graph's shortest cycle determines the minimum number of adversaries
required for the attack.
We conclude that removing short cycles from the network is a feasible countermeasure, albeit with considerable cost
towards the convergence speed of distributed protocols.

Our work sets the first step towards preventing reconstruction in the peer-to-peer setting as seen in multi-party
computation, and opens up multiple questions for future work.
Firstly, and most obviously, though we have found a sufficient criterion to determine reconstruction feasibility,
finding a criterion that is also necessary would allow using some graphs which our criterion currently forbids.
Secondly, our work is limited to a strictly syntactic notion of privacy, and does not protect linear relations between
data, which is required to protect against adaptive adversaries.
Thirdly, though our restriction to the summation operation is already sufficient to analyse decentralised learning, our
work could be extended to cover compositions with other operations, such as multiplication or comparison.
Finally, the addition of differentially private noise may further strengthen the provided level of privacy.